\documentclass[journal,onecolumn]{IEEEtran}
\usepackage{graphicx}
\usepackage{amsmath}
\usepackage{lipsum}
\usepackage{algcompatible}
\usepackage{algpseudocode}
\usepackage{amsthm}
\usepackage{amsfonts}
\usepackage{caption}
\usepackage{blkarray}
\usepackage{mathtools}
\usepackage{siunitx}
\usepackage{amssymb}
\usepackage{subcaption}
\usepackage{placeins}
\usepackage{xcolor}
\usepackage{lettrine}
\usepackage{multicol, blindtext}
\usepackage{bbm}
\usepackage{cite}
\usepackage{algorithm}
\usepackage{hyperref}
\usepackage{algpseudocode}
\usepackage{authblk}
\newtheorem{remark}{Remark}
\makeatletter

\newcommand{\Rmnum}[1]{\expandafter\@slowromancap\romannumeral #1@}

\newtheorem{theorem}{Theorem}
\newtheorem{corollary}{Corollary}
\newtheorem{proposition}{Proposition}
\newtheorem{assumption}{Assumption}
\newtheorem{definition}{Definition}
\newtheorem{lemma}{Lemma}

\makeatletter
\def\BState{\State\hskip-\ALG@thistlm}
\makeatother
\ifCLASSINFOpdf
\else
\fi

\begin{document}

\title{Minimizing the Age of Incorrect Information for Unknown Markovian Source}


\author{ Saad Kriouile}\author{Mohamad Assaad}
\affil{Laboratoire des Signaux et Syst\`emes, CentraleSup\'elec, Universit\'e Paris-Saclay,  91192 Gif sur Yvette, France}

\maketitle
\newcommand{\HRule}{\rule{\linewidth}{0.5mm}}

\begin{abstract}
The age of information minimization problems has been extensively studied in Real-time monitoring applications frameworks. In this paper, we consider the problem of monitoring the states of unknown remote source that evolves according to a Markovian Process. A central scheduler decides at each time slot whether to schedule the source or not in order to receive the new status updates in such a way as to minimize the Mean Age of Incorrect Information (MAoII). When the scheduler knows the source parameters, we formulate the minimization problem as an MDP problem. Then,  we prove that the optimal solution is a threshold-based policy, and we propose a low-complex algorithm that finds the optimal threshold.
When the source's parameters are unknown, the problem's difficulty lies in finding a strategy with a good trade-off between exploitation and exploration. Indeed, we need to provide an algorithm implemented by the scheduler that jointly estimates the unknown parameters (exploration) and minimizes the MAoII (exploitation). However, considering our system model, we can only explore the source if the monitor decides to schedule it. Then, applying the greedy approach, we risk definitively stopping the exploration process in the case where at a specific time, we end up with an estimation of the Markovian source's parameters to which the corresponding optimal solution is never to transmit. In this case, we can no longer improve the estimation of our unknown parameters, which may significantly detract from the performance of the algorithm. For that, we develop a new learning algorithm that gives a good balance between exploration and exploitation to avoid this main problem. Then, we theoretically analyze the performance of our algorithm compared to a genie solution by proving that the regret bound at time $T$ is $log(T)$. That implies that our solution converges to the optimal one at an efficient rate. 

Finally, we provide some numerical results to highlight the performance of our derived policy compared to the greedy approach.

\end{abstract}

\section{Introduction}

The remarkable growth of low-cost hardware has led to the emergence of real-time monitoring applications. In these systems, sensors are used to monitor events or environmental parameters, such as movement, temperature, humidity, and velocity. In order to have a timely reaction by the central entity, this later should receive from the sensor the recent status update of the remote source.

The main goal of these applications is accordingly to keep the monitor up-to-date by receiving the freshest information.


This concept of freshness is captured by the Age of Information (AoI), which is introduced for the first time in \cite{kaul2012real}. Since then, the AoI has become a hot research topic, and a considerable number of research works have been published on the subject  \cite{maatouk2020optimality,hsu2019scheduling,kadota2018scheduling,zou2019waiting,bedewy2016optimizing,bedewy2017age,bedewy2019age,maatouk2020status,sun2018age}.  

However, this metric doesn't consider the remote source's content. Precisely, it doesn't quantify the correctness of the information on the monitor's side.


This has been confirmed in \cite{jiang2019unified} where the authors establish that minimizing AoI gives a sub-optimal policy in minimizing the status error in remotely estimating Markovian sources.


To meet this goal, the authors in \cite{maatouk2020age} have designed a new metric dubbed Age of incorrect information AoII that captures both the freshness and the correctness of the information. Specifically, as long as the estimated state on the side of the receiver is different from the real state of the source, the AoII keeps growing by one per each time slot. While if the estimated state is equal to the actual state of the Markovian source, AoII goes to zero.
In \cite{maatouk2020age}, the authors have considered that the transmitter samples the source and decide whether to transmit the packet containing the useful information or not depending on the policy adopted. They have developed the optimal scheduling policy that minimizes the AoII under an energy constraint. 

Having said that, this metric is adopted in the Observable Markov Decision Problem framework in \cite{maatouk2020age}, i.e., in the case where the scheduler knows perfectly at each time the state of the Markovian source. To that extent, they assumed that the sensor knows exactly at each time slot the state of the source (always sampling the source) and performs, in addition, the scheduling task. This case may not be realistic since, in practice, the sensor has low energy and cannot perform both the sampling and the scheduling task at each time.

To deal with this issue, in \cite{kriouile2021minimizing}, the authors considered a Partially observable Markov Decision Problem framework where the monitor schedules the sensor to get the information of interest. To that extent, they proposed a slightly modified AoII metric to adapt it to the Partially observable Markov Decision Problem context. This metric estimates the value of AoII at each time $t$ on the side of the monitor as long as this later didn't receive the new status update regarding the information of interest yet. They considered a Markovian source with $N$ states and derived the low-complex and well-performing Whittle index policy in the multiple-sensors-one-receiver scenario. 
In order to get a simple increasing MAoII function with the age, they limited their analysis to the case where the transitioning probability of the Markovian source to another state is smaller than the remaining probability.     
However, in our case, we consider the pair-communication-scenario and extend the analysis done in \cite{kriouile2021minimizing} to the case where the transition probability could be higher than the remaining one. In this case, the difficulty lies on the fact that MAoII function is not an increasing function with the age but rather an oscillating function, which makes the investigations and analysis more challenging.
Moreover, when the source's parameters are unknown, we propose an online reinforcement algorithm that jointly estimates the parameters and minimizes the MAoII.

Regarding the learning aspect in the context of the Age of information minimization problem, one of the closest works to our proposed system model is \cite{tripathi2021online}. The authors proposed an online learning algorithm for the AoI minimization problem in this work. Thereby, they didn't take into account the information content of the remote source. Moreover, they considered that the successful transmission probability is equal to one. In order to ease the analysis regarding the bound of the regret function, they assumed that after $M$ time-slots, the packet is always transmitted. Leveraging these simplifying assumptions, they proposed an algorithm that gives a regret function bounded by the square root of $T$.

The paper \cite{fiez2018multi} shows the shortcomings of the standard UCB and $\epsilon$-greedy algorithms in the restless multi-armed bandit scheduling problem interacting with a correlated \textbf{markovian sources}. They proposed as alternative algorithms EpochUCB and EpochGreedy algorithms. Likewise, since the standard UCB algorithm is incompatible with the \textbf{MDP framework}, in \cite{auer2006logarithmic} and \cite{mete2021reward} and references therein, the authors developed round or episode-based learning algorithms using the optimistic approach for \textbf{Markov Decision Problems} such as UCRL, UCRL2, episodic Thompson sampling, RBMLE. In these algorithms, the policy corresponding to the estimated parameter at the beginning of the episode is applied during the entire episode to evaluate the average reward or the cost under this applied policy.   
Having said that, these papers used the notion of mixing time and considered that the source evolves under an \textbf{ergodic} Markov Chain in order to have a finite mixing time.
In contrast to these works, \cite{auer2008near} proposed a much weaker assumption than the ergodicity one: the Markov decision process (MDP) has a finite diameter. 
They defined the Diameter $D$ such that for any pair of states $s$, $s'$, there is a policy that moves from $s$ to $s'$ in at most $D$ steps (on average).
They evaluated for an undiscounted reinforcement learning problem the performance of a learning algorithm with respect to the optimal solution by analyzing the regret function.

Precisely, they proposed a reinforcement learning algorithm with total regret $O(D\sqrt{ST})$ after $T$ steps for any unknown MDP with $S$ states and diameter $D$. Though, in order to have a finite diameter, they considered a finite state space.
However, in our case, neither the MDP ergodicity nor the finite diameter assumptions are satisfied since the optimal solution may not give us an irreductible Markov process, and the state space is not finite.  
Moreover, we consider that the unknown parameters belong to a continuous set, while these mentioned works considered that the unknown parameters belong to a discrete set.
 

Another interesting work in the field of AoI that proposed an online learning algorithm is \cite{xiong2021learning}. In this work, the authors developed a learning augmented algorithm called UCB-whittle that estimates the parameters of the users' queues and applies the Whittle index policy at each epoch, considering the optimistic estimated parameters under the assumption that the markovian process is ergodic under the Whittle index policy. However, they considered that the system parameters belong to a discrete set which is already known by the algorithm. 

Unlike these works mentioned above, we suppose that the unknown parameter belongs to a continuous set, and we aim to estimate the transitioning probability of the Markovian process.
Moreover, in the opposite to these works, we don't use the episode-based approach where a given policy is applied for a fixed finite episode. Instead of that, we apply the estimated policy corresponding to the estimated parameter till we get a new estimation of the parameter in question. However, some estimations of the unknown parameter may provide us with scheduling policies that can be applied forever and eventually stop definitively the exploration process. This makes our problem more challenging since we need to avoid these policies while baring in mind at the same time that the optimal policy corresponding to the true parameter that we estimate could be itself among these typical policies. To that extent, we develop an algorithm that gives a good balance between the exploitation and exploration trade-off and that resolves this problem.

Specifically, our contributions can be summarized as follows:
\begin{itemize}
	\item When the parameters are known, and assuming that the remote source is volatile, we formulate the MAoII-based scheduling problem and provide the corresponding Bellman equation. Unlike \cite{kriouile2021minimizing}, we use a more complex and nontrivial analysis to establish that the optimal policy is a threshold-based policy since we get a non-monotone MAoII function regarding the age.
	\item We propose a low-complex algorithm that finds the optimal threshold policy.
	\item When the parameters are unknown, in contrast to \cite{tripathi2021online,fiez2018multi,mete2021reward,auer2006logarithmic,auer2008near, xiong2021learning}, in this work, we consider that the Markov Chain evolves under infinite state space and therefore we can not apply the episode-based approach.
Moreover, unlike \cite{tripathi2021online,fiez2018multi,mete2021reward,auer2006logarithmic,auer2008near, xiong2021learning}, where the authors consider that the estimation always occurs after a finite time, in our case, the estimation process depends on the policy applied. Indeed, it may no longer happen in some cases. For that, we propose an online reinforcement learning algorithm that matches this context and adjusts the exploration-exploitation trade-off.
	\item We compare between our solution and a genie algorithm and show that the regret bound is $O(Log(T))$ at time T, which implies that our solution converges to the optimal MAoII at an efficient convergence rate. 
	\item We provide numerical results that highlight the performance of our algorithm compared to the greedy policy.  
 \end{itemize}

\section{System Model}\label{sec:Syst_mod}
\subsection{Network description}\label{subsec:Net_descrip}
We consider in our paper one user that generates and send status updates about the process of interest to a central entity over unreliable channels. Time is considered to be discrete
and normalized to the time slot duration.
More specifically, the user observes an information process of interest $X(t)$ and at the request of the monitor, it samples the process $X(t)$ and send it to the monitor over an unreliable channel. Based on the last received update, the monitor constructs an estimate of the process, denoted by $\hat{X}(t)$.  
We suppose that the packet containing the information of interest, if it is successfully transmitted, will be instantaneously delivered to the monitor.  
In other words, if the monitor allows the user to transmit at time $t$, it receives the value of $X(t)$ at the same time $t$ if the packet is successfully transmitted. Therefore, it updates the estimate process as $\hat{X}(t)=X(t)$.
In any other case, namely when the user is not authorized to transmit or when the packet is unsuccessfully transmitted, the monitor keeps the same value at time slot $t$, specifically $\hat{X}(t)=\hat{X}(t-1)$.    
As for the unreliable channel, we suppose that at each time slot $t$, the probability of having successful transmission is $\rho$, and $1-\rho$ otherwise. Consequently, the channel realizations are independent and identically distributed
(i.i.d.) over time slots that we denote $c(t)$, i.e. $c(t)=1$ if the packet is successfully transmitted and $c(t)=0$ otherwise.

The next aspect of our model that we tackle is the nature of the process $X(t)$.
To that extent, the information process of interest $X(t)$ evolves under Markov chain. For that, we define the probability of remaining at the same state in the next time slot as $p$. Similarly, the probability of transitioning to another state is $r$. Denoting by $N$ the number of possible states of $X(t)$, then the following always holds:
\begin{equation}\label{eq:relation_p_r}
p+(N-1)r=1
\end{equation}
\begin{figure}[H]
\centering
\includegraphics[width=0.5\textwidth]{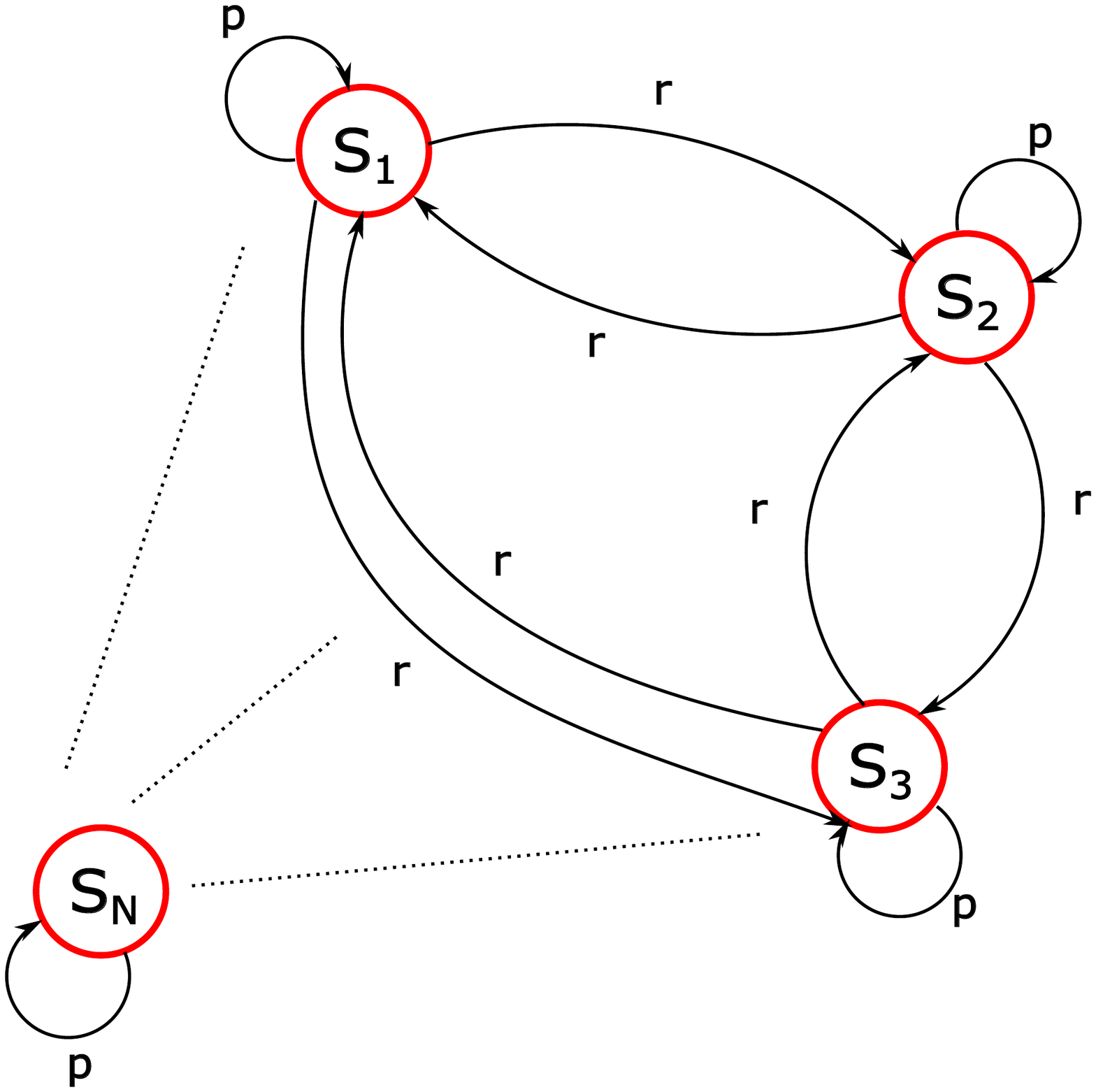}
\vspace{-4cm}
\caption{Illustration of process $X(t)$}
\end{figure}
Throughout this paper, we consider that the only unknown parameter by the monitor and the sensor is $r$. Moreover, we consider this following assumption.
\begin{assumption}\label{assump:volatility}
We consider a volatile source. In other words, the probability that the Markovian source transits to a different state is enough large. Explicitly, we assume that $(N-1)r \geq 4p$.  
\end{assumption}
In the sequel, we start first by deriving the optimal solution in the scenario where $r$ is considered to be known by the monitor, then we provide an algorithm that learns the parameter $r$ and minimizes simultaneously the objective function.
\subsection{MAoII metric}
In this paper, we study the mean age of incorrect information (MAoII) metric. The age of incorrect information has been introduced the first time in \cite{maatouk2020age}. Unlike the traditional AoI metric, this metric captures the freshness of \textbf{informative} updates. Specifically, if the monitor acquires the information about the process $X(t)$, as long as the state of the process $X(t)$ remains at the same state in the next time slots, the age of the incorrect information will not increase, since there is no new information unknown by the monitor. In \cite{maatouk2020age}, the authors presume that the scheduler has a perfect knowledge of the process at each time slot. While in our case, we consider that the monitor which plays the role of the scheduler, knows only the state of the last successively received packet. Accordingly, the explicit expression of MAoII metric is:
\begin{equation}
\delta_{MAoII}(t)=\mathbb{E}_{V}[(t-V(t)]
\end{equation}
where $V(t)$ is the last time instant such that $\mathbf{1}_{\{X(V(t))= \hat{X}(g(t))\}}=1$.
where $g(t)$\footnote{Considering our system model detailed in \ref{subsec:Net_descrip}, $g(t)$ refers also to the sampling time of the information of interest contained in the last successfully received packet} is the time-stamp of the last successfully received packet by monitor.
\begin{remark}
It is worth mentioning that, as it was explained in Section \ref{subsec:Net_descrip}, the reception of the successfully transmitted packet takes place at time slot $g(t)$. This means that $\hat{X}(g(t))=X(g(t))$.
\end{remark}
In order to use this metric effectively in a partially Observable
Markov Decision Process Problem, we need to take into consideration the markovian nature of the process $X(t)$. To that extent, we introduce in the next section the notion of the belief that represents the probability that $\hat{X}(t)$ is in the correct state.\\
 
\subsection{Metric evolution}\label{subsec:metrics_evolution}
In this section, we describe mathematically the evolution of the metric of interest depending on the system parameters and the action taken.
We denote by $d(t)$ the action prescribed to the user at time slot $t$ and by $a(t)$ the mean age of incorrect information function. 
To highlight the notion of the correctness, the monitor maintains a belief value $\pi(t)$ which is defined as the probability that the information state in the monitor, $\hat{X}(t)=\hat{X}(g(t))=X(g(t))$ at time $t$ being correct. Explicitly $\pi(t)=Pr(\hat{X}(t)=X(t))$. 
One can show that $\pi(t)$ evolves as follows:
\begin{lemma}\label{lem:pi_evolution}
\begin{equation}
 \pi(t+1)=\left\{
    \begin{array}{ll}
        1 \ \  \ \ \ \ \ \ \ \ if \ \  d(t+1)=1, c(t+1)=1 \\
        \pi(t)p+r(1-\pi(t)) \ \ \ else
    \end{array}
\right.
\end{equation}
\end{lemma}
\begin{IEEEproof}
See appendix A in \cite{kriouile2021minimizing}.
\end{IEEEproof}

According to the expression of MAoII given in section \ref{subsec:Net_descrip}, $(t-V(t))$ is a random variable that we denote $A(t)$ that satisfies:
\begin{lemma}\label{lem:random_variable}
\begin{align}
A(t)=\left\{
    \begin{array}{lll}
        0 & w.p & \pi(t)\\
        1 & w.p &  \pi(t-1).(1-p)\\
        2 & w.p &  \pi(t-2).(1-p).(1-r)\\
        3 & \cdots  &         \cdots\\
        \vdots & &\\
        t-g(t)-1& w.p & \pi(g(t)+1).(1-p)\\
                  &  & \  .(1-r)^{t-g(t)-2} \\
        \\          
        t-g(t)& w.p & (1-p).(1-r)^{t-g(t)-1} 
    \end{array}
\right.
\end{align}
\end{lemma}
\begin{IEEEproof}
See appendix B in \cite{kriouile2021minimizing}.
\end{IEEEproof}

Therefore, the mean of the age of the incorrect information at slot $t$ equals to the mean of $A(t)$, i.e.
\begin{align}
n(t)=&\mathbb{E}[A(t)] \nonumber \\
=&\sum_{k=0}^{t-g(t)-1} k(1-p)(1-r)^{k-1} \pi(t-k) \nonumber \\
&+(t-g(t)).(1-p).(1-r)^{t-g(t)-1} \nonumber \\
=&\sum_{k=1}^{t-g(t)} (t-g(t)-k)(1-p)(1-r)^{t-g(t)-k-1} \pi(g(t)+k)\nonumber \\
&+ (t-g(t)).(1-p).(1-r)^{t-g(t)-1} 
\end{align}
One can establish that for all $t$, using definition of $g(t)$, $\pi(g(t))=1$. Hence, according to the evolution of $\pi(\cdot)$ in Lemma \ref{lem:pi_evolution}, for all $k \leq t-g(t)$, $\pi(g(t)+k)$ depends only on $k$. More precisely, we have that for each $k \leq t-g(t)$, $\pi(g(t)+k)=\pi_k$ where $\pi_k$ is a sequence defined by induction as follows:
\begin{equation}
 (\pi_k)=\left\{
    \begin{array}{ll}
        \pi_0=1\\
        \pi_{k+1}=p \pi_k+r(1-\pi_k) & if \ \ k \geq 0
    \end{array}
\right.
\end{equation}
In light of that fact, we have that:
\begin{equation}
n(t)=\sum_{k=0}^{t-g(t)} (t-g(t)-k)(1-p)(1-r)^{t-g(t)-k-1} \pi_k
\end{equation}
We conclude that $n(t)$ depends on $t-g(t)$. Therefore, we let $n(t)\overset{\Delta}{=}n(t-g(t))$.\\

To that extent, at time slot $t+1$, if the user is scheduled and the packet is successively transmitted, then $g(t+1)=t+1$. Accordingly, at time slot $t+1$, MAoII equals to $n(t+1-g(t+1))=n(0)$.   
If the user is not scheduled or if the packet is not successively transmitted, then $g(t+1)=g(t)$. Therefore, MAoII will transit to $n(t+1-g(t+1))=n(t-g(t)+1)$. Based on this and denoting $j(t)$ the index such that $n(j(t))$ is the value of MAoII at time slot $t$, MAoII will transit to the value $n(j(t)+1)$ at time instant $t+1$. To sum up, the evolution of MAoII can be summarized as follows:     
\begin{equation}
 a(t+1)=\left\{
    \begin{array}{ll}
        n(0) & if d(t+1)=1, c(t+1)=1 \\
        n(j(t)+1) & else
    \end{array}
\right.
\end{equation}
where $a(t)=n(j(t))$.
\section{Problem formulation}\label{sec:prob_form}
In this paper, we consider that if the user is scheduled, an additional cost should be paid due to the energy consumption during the transmission process. 
To that extent, we let $C(t)=a(t)+\lambda d(t)$ be the penalty function at the central entity of the user at time slot $t$ where $\lambda >0$ and $d(t)=1$ if the packet if the user is scheduled and $0$ otherwise. Our aim is to find a scheduling policy that decides to whether schedule the user or not in a such way to minimize the total expected average penalty function. A scheduling policy $\phi$ is defined as a sequence of actions $\phi=(d^{\phi}(0),d^{\phi}(1),\ldots)$ where $d^{\phi}(t)=1$ if the user is scheduled at time $t$, and $d^{\phi}(t)=0$ otherwise. 
Denoting by $\Phi$, the set of all causal scheduling policies, then
our scheduling problem can be formulated as follows:
\begin{equation}\label{eq:problem_formulation}
\setlength{\belowdisplayskip}{0pt} \setlength{\belowdisplayshortskip}{0pt}
\setlength{\abovedisplayskip}{0pt} \setlength{\abovedisplayshortskip}{0pt} 
\begin{aligned}
& \underset{\phi\in \Phi}{\text{minimize}}
& & \lim_{T\to+\infty} \text{sup}\:\frac{1}{T}\mathbb{E}^{\phi\in \Phi}\Big(\sum_{t=0}^{T-1}C^{\phi}(t)|C(0)\Big)\\
\end{aligned}
\end{equation}

\subsection{Structural results}
The problem in (\ref{eq:problem_formulation}) can be viewed as an infinite horizon average cost Markov decision process that is defined as follows:
\begin{itemize}
\item \textbf{States}: The state of the MDP at time $t$ is the MAoII function $a(t)$. 
\item \textbf{Actions}: The action at time $t$, denoted by $d(t)$, specify if the user is scheduled (value $1$) or not (value $0$).
\item \textbf{Transitions probabilities}: The transitions probabilities between the different states.
\item \textbf{Cost}: The instantaneous cost of the MDP, $C(a(t),d(t))$, be equal to $a(t)+\lambda d(t)$.
\end{itemize}
The optimal policy $\phi^*$ of Problem \eqref{eq:problem_formulation} can be obtained by solving the following Bellman equation for each state $a$:
\begin{align}
\theta &+ V(a) \nonumber \\ 
&=\min_{d\in\{0,1\}}\big\{a+\lambda d+\sum_{a'\in A^a }\Pr(a\rightarrow a'|d)V(a')\big\} 
\label{eq:bellman_general}
\end{align}
where $\Pr(a\rightarrow a'|d)$ is the transition probability from state $a$ to $a'$ under action $d$, $\theta$ is the optimal value of the problem, $V(a)$ is the differential cost-to-go function and $A$ is the set of states of the MAoII metric.
\subsection{Discounted Cost approach}
In order to examine the structure of the optimal solution of Problem \eqref{eq:problem_formulation}, we adopt the discounted cost approach. This approach is widely used in the framework of MDP problems (e.g. \cite{liu2010indexability,larranaga2017asymptotically}). Precisely, it consists of introducing a discounted factor $\beta$ in the cost function and deriving its optimal solution. After that, exploiting the fact $\lim V_{\beta}(\cdot)=V(\cdot)$ when $\beta$ goes to $1$ under some conditions that we will prove later, we conclude that the optimal solution of \eqref{eq:problem_formulation} follows the same structure as the discounted one. 

To that end, we introduce the Bellman equation of the discounted cost problem as follows:
\begin{align}
V_{\beta}(a) \nonumber \\ 
&=\min_{d\in\{0,1\}}\big\{a+\lambda d+\sum_{a'\in A }\Pr(a\rightarrow a'|d)\beta V_{\beta}(a')\big\} 
\label{eq:bellman_general}
\end{align}
where $\beta$ is strictly less than $1$. 
There exist several numerical algorithms that are developed to solve (\ref{eq:bellman_general}), such as the value iteration algorithm. This later consists first of updating per each iteration the value function $V_{\beta}^t(.)$ following the recurrence relation for each state $a$:
\begin{align}\label{eq:bellman_equation_time_t}
V^{t+1}_{\beta}(a) \nonumber\\
&=\min_{d\in\{0,1\}}\big\{a+\lambda d+\sum_{a'\in A }\Pr(a\rightarrow a'|d)\beta V^t_{\beta}(a')\big\}  
\end{align}
Given that $V^0_{\beta}(.)=0$, we compute $V_{\beta}(.)$ exploiting the fact that $\underset{t \rightarrow +\infty}{\text{lim}} V^t_{\beta}(a)=V_{\beta}(a)$ (see \cite[Chapter~8.5]{puterman2014markov}). The main shortcoming of this algorithm is that it requires high memory and computational complexity. To overcome this complexity, rather than computing the value of $V_{\beta}(.)$ for all states, we limit ourselves to study the structure of the optimal scheduling policy by exploiting the fact that $\underset{t \rightarrow +\infty}{\text{lim}} V^t_{\beta}(a)=V_{\beta}(a)$. In that way, we show that the optimal solution of Problem \eqref{eq:bellman_general} is a threshold-based policy: 
\begin{definition} 
A threshold policy is a policy $\phi \in \Phi$ for which there exists $n$ such that when the current state $a < n$, the prescribed action is $d^- \in \{0,1\}$, and when $ a \geq n$, the prescribed action is $d^+ \in \{0,1\}$ while baring in mind that $d^- \neq d^+$.
\end{definition}
To that extent, we show that the optimal policy of \eqref{eq:bellman_general} is a threshold based policy. 
For that purpose, we specify first the states space $A$, then we provide the expression of the corresponding Bellman equation \eqref{eq:bellman_general}. After that, we establish our desired result.
According to Section \ref{subsec:metrics_evolution}, $a(t)$ evolves in the state space:
\begin{equation}
A=\{a_{j}: j \geq 0, a_{j}=\sum_{k=0}^{j} k(1-p)(1-r)^{k-1} \pi_{j-k}\}
\end{equation}
Therefore, the expression of Bellman equation at state $a_{j}$ 
\begin{align}\label{eq:bellman_equation_discounted}
V(a_{j})=\min\big\{&a_{j}+\beta V_{\beta}(a_{j+1});\nonumber \\
&a_{j}+\lambda+\rho \beta V_{\beta}(a_{0})+(1-\rho)\beta V_{\beta}(a_{j+1})\big\} 
\end{align}

\begin{theorem}\label{theo:threshold_policy_discounted}
The optimal solution of the discounted problem in (\ref{eq:bellman_equation_discounted}) is an increasing threshold policy. Explicitly, there exists $a_n$ such that when the current state $a_{j} < a_n$, the prescribed action is a passive action, and when $ a_j \geq a_{n}$, the prescribed action is an active action.
\end{theorem}
\begin{IEEEproof}
The proof can be found in Appendix \ref{app:theo:threshold_policy_discounted}.
\end{IEEEproof}
Leveraging the Theorem \ref{theo:threshold_policy_discounted}, we prove that the optimal solution of the original problem is threshold policy as well. For that, we introduce the following theorem, which is given in \cite[Section~V, Theorem 2.2]{ross2014introduction}, that links between the discounted value function and that of the original problem.
\begin{theorem}\label{theo:link_discounted_original}
If there exists a constant $K$ such that for all $i$ and $0\leq \beta <1$, we have $|V_{\beta}(a_i)-V_{\beta}(a_0)|<K$, then:\\
For some sequence $\beta_n \rightarrow 1$, $V(a_i)=\underset{n \rightarrow +\infty}{\lim}[V_{\beta_n}(a_i)-V_{\beta_n}(a_0)]$
\end{theorem}
According to this theorem, it follows that $V(\cdot)$ inherits the structural form of $V_{\beta}(\cdot)$. That means that as $V_{\beta}(\cdot)$ is increasing with $a_i$, $V(\cdot)$ is also increasing with $a_i$.
Having said that, we need to check first that the condition given in Theorem \ref{theo:link_discounted_original} is satisfied for our specific system model.
\begin{theorem}\label{theo:threshold_policy_original_condition}
There exists a constant $K$ such that for all $i$ and $0\leq \beta <1$, we have $|V_{\beta}(a_i)-V_{\beta}(a_0)|<K$.
\end{theorem} 
\begin{proof}
See appendix \ref{app:theo:threshold_policy_original_condition}
\end{proof}
\begin{theorem}
The optimal solution of the original problem in (\ref{eq:problem_formulation}) is an increasing threshold policy. Explicitly, there exists $a_n$ such that when the current state $a_{j} < a_n$, the prescribed action is a passive action, and when $ a_j \geq a_{n}$, the prescribed action is an active action.
\end{theorem}
\begin{IEEEproof}
From Theorems \ref{theo:link_discounted_original} and \ref{theo:threshold_policy_original_condition}, it follows that the optimal solution of the original problem inherits the structure of the optimal solution of the discounted cost problem. As consequence, the optimal solution of the original problem is a threshold based policy.
\end{IEEEproof}
\subsection{Closed-form expression of the optimal solution}
As we have proved in the previous section, the optimal solution of \eqref{eq:problem_formulation} is threshold-based policy. Based on that, we provide in this section the algorithm that allows us to find the exact threshold policy for a given $\lambda$. Although we know the structure of the optimal policy, we still have to determine the exact threshold. To that end, we should derive a simple closed-form expression of Problem \eqref{eq:problem_formulation} that we can investigate easily. Indeed as the optimal solution of \eqref{eq:problem_formulation} is a threshold policy, we derive the steady-state form of the problem in (\ref{eq:problem_formulation}) under a given threshold policy $n$. Explicitly:
\begin{equation}
\begin{aligned}
& \underset{n\in \mathbb{N}^*}{\text{minimize}} 
& & \overline{a^{n}}+\lambda\overline{d^n}
\end{aligned}
\label{thresholdobjective}
\end{equation}
where $\overline{a^{n}}$ is the average value of MAoII, and $\overline{d^n}$ is the average active time under threshold policy $n$. Specifically: 
\begin{align}
\overline{a^{n}}&=\lim_{T\to+\infty} \text{sup}\:\frac{1}{T}\mathbb{E}^{n}\Big(\sum_{t=0}^{T-1}a(t)|a(0),tp(n)\Big)\label{eq:average_age}\\
\overline{d^n}&=\lim_{T\to+\infty} \text{sup}\:\frac{1}{T}\mathbb{E}^{n}\Big(\sum_{t=0}^{T-1}d(t)|a(0),tp(n)\Big)\label{eq:average_active_time}
\end{align}
where $tp(n)$ denotes the threshold policy $n$.
With the intention of computing $\overline{a^{n}}$ and  $\overline{d^n}$, we derive the stationary distribution of the Discrete Time Markov Chain, DTMC that represents the evolution of MAoII under threshold policy $n$. 
\begin{proposition}\label{prop:stationary_distribution}
We distinguish between two cases:
\begin{itemize}
\item $1-r \geq |p-r|$:
For a given threshold $a_n$, the DTMC admits $u_n(a_{i})$ as its stationary distribution:
\begin{equation}
 u_n(a_{i})=\left\{
    \begin{array}{ll}
        \frac{\rho  }{n\rho  +1} & \text{if} \ 0 \leq i \leq n  \\
        (1-\rho  )^{i-n} \frac{\rho  }{n\rho  +1  } & \text{if} \ i \geq n+1\\
    \end{array}
\right.
\end{equation}
\item $1-r < |p-r|$:
For a given threshold $a_n$, the DTMC admits $u_n(a_{i})$ as its stationary distribution:
\begin{itemize}
\item If $n=2k+1$, then DTMC doesn't admits a stationary threshold. In this case, all the states of the DTMC are transients and in the long-term, the threshold policy will be to not transmit for all states of DTMC.
\item If $n=2k$, then:
\begin{equation}
 u_n(a_{i})=\left\{
    \begin{array}{ll}
        \frac{(1-\rho)^{[i/2]}\rho}{2-(1-\rho)^{n/2}} & \text{if} \ 0 \leq i \leq n-1  \\
         \frac{(1-\rho  )^{i-n/2} \rho}{2-(1-\rho)^{n/2}} & \text{if} \ i \geq n+1\\
    \end{array}
\right.
\end{equation} 
\end{itemize}
\end{itemize}
\label{stationarydistribution}
\end{proposition}
\vspace{-10pt}
\begin{IEEEproof}
The proof can be found in Appendix \ref{app:prop:stationary_distribution}.
\end{IEEEproof}

By exploiting the above results, we can now proceed with finding a closed-form of the average cost under any threshold policy.

\begin{proposition}
Under a threshold policy $n$, the average cost denoted by $\overline{a^{n}}$ is equal to:
\begin{itemize}
\item $1-r\geq |p-r|$:
\begin{align}
\overline{a^{n}}=&\frac{\rho }{n\rho  +1}[\frac{n(N-1)}{Nr}-\frac{(1-Nr)^{n+2}}{(Nr)^2}\nonumber\\
&+\frac{(1-r)^{n+2}}{r^2}+\frac{(1-\rho)(1-Nr)^{n+2}}{Nr(1-(1-\rho)(1-r))}-\nonumber\\
&\frac{(1-\rho  )(1-r)^{n+2}}{r(1-(1-\rho)(1-r))}+C]
\end{align}
\item $1-r < |p-r|$ ($n$ is an even number):
\begin{align}
&\overline{a^{n}}\nonumber\\
=& \frac{N-1}{Nr}\nonumber\\
+&\frac{\rho(p-r) }{(2-(1-\rho)^{n/2})Nr}\nonumber\\
\times & [(2-Nr)(1-(1-\rho)(p-r))-\rho(1-Nr)^{n+1}(1-\rho)^{n/2}]\nonumber\\
\times & \frac{1}{(1-(1-\rho)(p-r)^2)(1-(1-\rho)(p-r))} \nonumber\\
-&\frac{\rho(1-r)}{(2-(1-\rho)^{n/2})r}\nonumber\\
\times & [(2-r)(1-(1-\rho)(1-r))-\rho(1-r)^{n+1}(1-\rho)^{n/2}]\nonumber\\
\times & \frac{1}{(1-(1-\rho)(1-r)^2)(1-(1-\rho)(1-r))}
\end{align}
\end{itemize}
where $C=\frac{(1-Nr)^2}{(Nr)^2}-\frac{(1-r)^2}{r^2}+\frac{(N-1)(1-\rho )}{Nr\rho}$.
\end{proposition}
\begin{IEEEproof}
By leveraging the results of Proposition \ref{prop:stationary_distribution} and using the expression of $a_{j}$ for $j \geq 0$, by definition of $\overline{a^{n}}$ given in \eqref{eq:average_age}, we get after algebraic manipulations the desired results.
\end{IEEEproof}

\begin{proposition}
The active average time denoted by $\overline{d^n}$ is equal:
\begin{itemize}
\item $1-r \geq |p-r|$:
\begin{align}
\overline{d^n}=&\frac{1}{n\rho +1}
\end{align}
\item $1-r < |p-r|$( $n$ is an even number):
\begin{align}
\overline{d^n}=&\frac{1}{2-(1-\rho)^{n/2}}
\end{align}
\end{itemize}
\end{proposition}
\begin{IEEEproof}
Likewise, exploiting the results in Proposition \ref{prop:stationary_distribution} and according to the expression \eqref{eq:average_active_time}, we obtain the desired results.
\end{IEEEproof}

Since we found the steady state form of Problem \eqref{eq:problem_formulation}, our objective will be to find out the threshold $n$ that minimizes $a^n+\lambda d^n$. A brute-force scheme will be to compare between the $a^n+\lambda d^n$ for different values of threshold $n$. However, this comparison process will be endless since $n$ belongs to infinite set, as MAoII evolves within infinite state space. Thereby, this classical approach falls short considering our system settings. To overcome this issue we proceed with more analysis in order to have a low-complex algorithm that allows us to determine the optimal threshold for a given $\lambda$ without the need of comparing between the costs function for different thresholds policies.
For that purpose, we describe the optimal threshold as a function of $\lambda$.
We start first by the case where the $1-r \geq |p-r|$
\subsubsection{$1-r \geq |p-r|$}     
We first define the sequence $\lambda(a_{n})$
as follows:
\begin{definition}\label{def:intersection_poin_lambda}
$\lambda(a_{n})$ is the intersection point between $\overline{a^n}+\lambda\overline{d^n}$ and $\overline{a^{n+1}}+\lambda \overline{d^{n+1}}$. Explicitly:
\begin{equation}
\lambda(a_{n})=\frac{\overline{a^{n+1}}-\overline{a^{n}}}{\overline{d^n}-\overline{d^{n+1}}}
\end{equation}
\end{definition}
\begin{theorem}\label{theo:expression_threshold_policy_in_function_lambda}
The optimal threshold policy of Problem \eqref{eq:problem_formulation} satisfies:
\begin{itemize}
\item If $\lambda \leq \lambda(a_{0})$, then the optimal threshold is $a_{0}$
\item If $\lambda(a_n) < \lambda \leq \lambda(a_{n+1})$, then the optimal threshold is $a_{n+1}$
\item If $\lambda \geq \underset{k \Rightarrow +\infty}{\lim} \lambda(a_k)=\lambda_l$, then the optimal threshold is infinite.
\end{itemize}
\end{theorem}
\begin{proof}
See appendix \ref{app:theo:expression_threshold_policy_in_function_lambda}
\end{proof}
\subsubsection{$1-r < |p-r|$}     
For this case, as was indicated in Proposition \ref{prop:stationary_distribution}, the set of the eventual threshold policies is $\{2n: n \in \mathbf{N}\} \cup \{+\infty\}$. To that extent we define $\lambda(a_{2n})$ as follows:
\begin{definition}
$\lambda(a_{2n})$ is the intersection point between $\overline{a^{2n}}+\lambda\overline{d^{2n}}$ and $\overline{a^{2(n+1)}}+\lambda \overline{d^{2(n+1)}}$. Explicitly:
\begin{equation}
\lambda(a_{2n})=\frac{\overline{a^{2(n+1)}}-\overline{a^{2n}}}{\overline{d^{2n}}-\overline{d^{2(n+1)}}}
\end{equation}
\end{definition}
 \begin{theorem}\label{theo:expression_threshold_policy_in_function_lambda_second_case}
The optimal threshold policy of Problem \eqref{eq:problem_formulation} satisfies:
\begin{itemize}
\item If $\lambda \leq \lambda(a_{0})$, then the optimal threshold is $a_{0}$
\item If $\lambda(a_{2n}) < \lambda \leq \lambda(a_{2(n+1)})$, then the optimal threshold is $a_{2n}$
\item If $\lambda \geq \underset{k \Rightarrow +\infty}{\lim} \lambda(a_k)$, then the optimal threshold is infinite.
\end{itemize}
 \end{theorem}
\begin{proof}
See appendix \ref{app:theo:expression_threshold_policy_in_function_lambda_second_case}
\end{proof}
Based on the above theorems, we provide in the following the algorithm that allows us to find the optimal threshold policy. We focus on the the first case where $1-r\geq |p-r|$ since the algorithm is almost the same for the two cases.

\begin{algorithm}[H]
\caption{Optimal Threshold Policy}\label{euclid}
\begin{algorithmic}[1]
\State Init. $t=0$
\State Init. $n=0$
\State Init. $p=1$
\State If $\lambda \leq \lambda(a_0)$:
		$n^*=0$
\State If $\lambda \geq \lambda_l$:
		$n^*=+\infty$
\State Else:
\State \hspace{1 cm} While $p==1$:
\State  \hspace{1 cm} \hspace{1 cm} $n=n+ \lfloor \alpha(\lambda-\lambda(n))\rfloor$
\State \hspace{1 cm} \hspace{1 cm}		If $\lambda(n)<  \lambda \leq \lambda(n+1)$
\State 	  \hspace{1 cm} \hspace{2 cm}	 $p=0$
\State 	   \hspace{1 cm} \hspace{2 cm}	$n^*=n+1$
\State Return $n^*$
\end{algorithmic}
\end{algorithm}
where $\lfloor x \rfloor$ is the integer part of $x$.
\begin{remark}
To ensure the convergence of the above algorithm, we need to have $\alpha< \frac{k}{\lambda(a_{n+k})-\lambda(a_n)}$ for all integer $n$ and $k$. This is satisfied when $\alpha < \frac{1}{(1-r)^2-(p-r)^2}$.
\end{remark}
\section{Unknown source parameters}\label{sec:unknwon_source}
In this section, we consider that we don't known the transition probabilities of the source, namely $p$ and $r$. Our aim is to implement an algorithm that optimally learns these parameters. In other words, we derive an algorithm that learns the parameters in question and optimize in the same time the average age of incorrect information. 
We consider that $N >2$, that implies that $1-r \geq |p-r|$, i.e., the function $\lambda(.)$ defined in Definition \ref{def:intersection_poin_lambda} is increasing with $n$.
We consider that $\rho$ is known as well as the number of states referred by $N$.
Since there is an equation that links between $N$, $p$ and $r$, then the only unknown parameter that we should estimate is $r$. We note that the real value of $r$ is $r^*$.
To that extent, our objective through this section is to find out, for a fixed time horizon $T$, the suitable algorithm that minimizes the gap between the total cost and the one under a genie policy (when the parameters are known). In other words, our goal is to minimize the following regret function:

\begin{align}\label{eq:regret_function_expression}
R_{\pi}(T)&=E\Big [ \sum_{t=0}^T C(a^{\pi}(t),d^{\pi}(t))\Big ]-TC^*
\end{align}
where $C^*=\bar{a^{n^*}}+\lambda \bar{d^{n^*}}$, $n^*$ is the optimal threshold corresponding to the true parameter $r^*$, $\bar{a^{n^*}}$ is the average MAoII under the threshold policy $n^*$, and $\bar{d^{n^*}}$ is the average active time under threshold policy $n^*$ as defined in the previous Section.\\  

In the sequel, we provide our algorithm and we show that under it, the regret function is less than $O(log(T))$ under our proposed algorithm. By doing so, we could say that the average cost under our algorithm converges to the optimal cost as $\frac{R_{\pi}(T)}{T}$ approach to $0$.
For that, we define the function $\lambda(n,r)$ as $\lambda(a_n)$ (defined in Definition \ref{def:intersection_poin_lambda}) when the transition probability is equal to $r$. Accordingly, $\lambda(n,r)=\frac{(N-1)(N+1-N r) \rho}{N^2 r^2}+(1-Nr)^{n+2}(n\rho+1+\frac{\rho}{N r}) \nonumber \times [\frac{1-(1-\rho)(1+(N-1)r)}{N r (1-(1-\rho)(1-r))} ]-(1-r)^{n+2}(n\rho+1+\frac{\rho}{r})\times [\frac{\rho}{ r (1-(1-\rho)(1-r))}] $. 
To that extent, we study the monotony of $\underset{n \rightarrow +\infty}{\lim}\lambda(r,n)=\lambda(r)$ in function of $r$. Indeed, it is clear from the expression of $\lambda(r)=\frac{(N-1)(N+1-N r) \rho}{N^2 r^2}$, that this later tends to $+\infty$ if $r \rightarrow 0$ and decreases as $r$ grows.

Given the algorithm \ref{euclid} that describes the optimal threshold policy in function of $\lambda$, for the values $r$ such that $\lambda(r)$ is less than $\lambda$, the optimal solution is the infinite threshold.
While, for the values $r$ where $\lambda(r) > \lambda$, as for a fixed $r$, $\lambda(n,r)$ tends to $\lambda(r)$ when $n$ grows, then there exits surely $n_r$ such that $\lambda(n_r,r) < \lambda \leq \lambda(n_r+1,r) $. Hence, according to Algorithm \ref{euclid}, when $\lambda(r) > \lambda$, the optimal threshold in finite 

Therefore, baring in mind that $\lambda(r)$ is decreasing with $r$, then there exists $r_l$ such that for any $r\geq r_l$ ($\lambda(r) \leq \lambda$), the optimal threshold is infinite and for all $r<r_l$ ($\lambda(r) > \lambda$), the optimal threshold is finite. Having said that, to design an algorithm that allows us to explore enough the unknown parameter $r$, we need to avoid applying the infinite threshold at any time, otherwise, we will not be able to explore $r$ any more since we will never transmit. On the other hand, in the case where the real value of $r$ is greater than $r_l$, then the optimal threshold is finite. Thus, applying always finite threshold policy will be sub-optimal. To deal with this issue, we consider that after large enough time, after we obtain a good estimation of $r$, we decide whether to apply the infinite threshold or to keep estimating our value $r$.    
We note that throughout this section, we consider a sufficient large $T$ which is known by the scheduler. Moreover we consider that $r^*\neq r_l$.
In the sequel, we provide our detailed algorithm.

Before presenting our algorithm, for ease of understanding, we give some useful definitions:
\begin{definition}
\begin{itemize}
\item $t_i$ refers to the time-stamp of the i-th successful transmitted packet under our proposed algorithm.
\item $r_i$ refers to the estimated transition probability at time $t_i$.
\item $n(r_i)$ refers to the optimal threshold policy when the transition probability $r$ is equal to $r_i$.
\item $p(t)$ refers to the realization of the transition at time $t$. In other words, if $p(t)=0$, then at time $t+1$ the source remains at the same state, otherwise, it transits to another state.   
\item $N(t)$ counts the number of times the sensor has successfully transmitted the information of interest from time $0$ till time $t$. Specifically, $N(t_i)=i$. 
\end{itemize}
\end{definition} 

\begin{algorithm}
\caption{Iterative algorithm for threshold optimal policy}\label{euclid2}
\algblockdefx[Name]{Begin}{End}[1]{\textbf{#1}}{}
\algtext*{End}
\begin{algorithmic}[1]
\Begin {At $t_0=0$, at the side of the monitor:}
\State \textbf{Init} $r_0=0$.
\State Apply threshold $0$ until receiving the new update at
\State time $t_1$. 
\End
\Begin {At $t_1$:}
     \Begin {At the side of the sensor:}
     \State At the end of $t_{1}$: the sensor stores the value $p(t_{1})$.
     \End
     \Begin {At the side of the monitor:}
     \State Apply threshold $0$ till $t_2$.
     \End 
   \End  
\State \textbf{Init} $i=2$
\While {($r_{i-1}+\sqrt{log(T)/i-1} \geq r_l$ \textbf{And} $r_{i-1} < r_l$) \textbf{Or} \hspace{1cm} ($r_{i-1}-\sqrt{log(T)/i-1} < r_l$ \textbf{And} $r_{i-1} \geq r_l$)} 
	\State  At time $t_{i}$:
	\Begin {At the side of the sensor:}
		\State Transmit along with the information $X(t_{i})$, the   				\State realization of $p(\cdot)$ at time $t_{i-1}$ stored at  
		\State time $t_{i-1}$ to the monitor.
		\State At the end of $t_{i}$: the sensor stores the value 					$p(t_{i})$.
    \End
	\Begin {At the side of the monitor:}
		\State Update the estimator of $r$ as follows:
		\State $r_{i-1}=\frac{i-2}{i-1}r_{i-2}+\frac{1}{N-1}\frac{p(t_{i-1})}{i-1}$.
		\State Apply the threshold policy $0$.
	\End	  
	\State $i++$
\EndWhile
\If {$r_{i-1}+\sqrt{log(T)/i-1} < r_l$} 
	\State Consider $r^*$ strictly less than $r_l$.
	\While {$r_{i-1}+2\sqrt{log(T)/i-1} \geq r_l$}
		\State At time $t_{i}$:
		\Begin {At the side of the sensor:}
			\State Transmit along with the information $X(t_{i})$, the 					\State realization of $p(\cdot)$ at time $t_{i-1}$ stored at                           			\State time $t_{i-1}$ to the monitor.
			\State At the end of $t_{i}$: the sensor stores the value 					\State $p(t_{i})$.
    	\End
		\Begin	{At the side of the monitor:}	
			\State Update the estimator of $r$ as follows:
			\State $r_{i-1}=\frac{i-2}{i-1}r_{i-2}+\frac{1}{N-1}\frac{p(t_{i-1})}{i-1}$.
			\State Apply the threshold policy $0$.
		\End	  
		\State $i++$
	\EndWhile
	\State At time $t_{i}$:
	\Begin {At the side of the sensor:}
		\State Transmit along with the information $X(t_{i})$, the 					\State realization of $p(\cdot)$ at time $t_{i-1}$ stored at time 			\State $t_{i-1}$ to the monitor. 
		\State At the end of $t_{i}$: the sensor stores the value $p(t_{i})$.
	\End
	\Begin {At the side of the monitor:}
		\State Update the estimator of $r$ as follows: 
		\State $r_{i-1}=\frac{i-2}{i-1}r_{i-2}+\frac{1}{N-1}\frac{p(t_{i-1})}{i-1}$.
		\State \textbf{if} $r_{i-1}<r_l$ \textbf{then}: Apply the optimal threshold 
		\State denoted by $n(r_{i-1})$ using algorithm \ref{euclid}.
		\State \textbf{else}: Apply the zero threshold.	 
	\End
\EndIf
\If {$r_{i-1}-\sqrt{log(T)/i-1} \geq r_l$}
	\State Consider $r^*$ greater than $r_l$.
	\State At the side of the monitor: Apply the infinite
	\State threshold policy.
\EndIf
\end{algorithmic}
\end{algorithm}
\begin{theorem}\label{theo:bound_regret}
There exists a constant $K$ independent of $T$ such that: 
\begin{equation}
R_{\pi}(T) \leq K log(T)
\end{equation} 
where $\pi$ corresponds to the policy with respect to the algorithm \ref{euclid2}.
\end{theorem}
The result of this Theorem means, on the one hand, that the average cost under our proposed algorithm converges to the optimal one when $T$ grows as $\frac{log(T)}{T}$ goes to zero. On the other hand, it means that the convergence rate to the optimal solution is $O(\frac{log(T)}{T})$, which is a very good rate.
\section{Proof of Theorem \ref{theo:bound_regret}}
Before proving the theorem, we give some preliminaries results.

As we can notice in Algorithm \ref{euclid2}, we select at each time $t_i$, the mean estimator of $r^*$. In fact, we have $r_i=\frac{1}{N-1}\sum_{k=1}^i \frac{p(t_k)}{i}$. The expectation of $p(t_k)$ is exactly $(N-1)r^*$. Therefore according to Hoeffding inequality, we have:
\begin{lemma}\label{lem:hoeffding_inequality}
For any $\epsilon>0$, we have that:
\begin{equation}
P(|r_i-r^*|>\epsilon/\sqrt{i}) \leq 2 \exp (-2(N-1)^2\epsilon^2)
\end{equation}
\end{lemma}
\begin{corollary}\label{cor:probability_G_T}
By letting $G(T)$ be the event: 
\begin{equation}
\underset{i=[1,\cdots,T]}{\bigcap} \{ |r_i-r^*| \leq \sqrt{\frac{log(T)}{i}} \}
\end{equation}
Then: 
$$P(G(T)) \geq 1- \frac{2}{T}$$
\end{corollary}
\begin{proof}
See appendix \ref{app:cor:probability_G_T}.
\end{proof}
Leveraging the Corollary above, we obtain:
\begin{align}\label{regret_ineq_1}
R_{\pi}(T) \leq &E(\sum_{t=0}^T |C(a^{\pi}(t),d^{\pi}(t))-C^*|) \nonumber \\
=&E(\sum_{t=0}^T |C(a^{\pi}(t),d^{\pi}(t))-C^*||G(T))P(G(T))\nonumber \\
+&E(\sum_{t=0}^T |C(a^{\pi}(t),d^{\pi}(t))-C^*||\bar{G}(T))P(\bar{G}(T)) \nonumber \\
\leq &E(\sum_{t=0}^T |C(a^{\pi}(t),d^{\pi}(t))-C^*||G(T))+2M \nonumber \\
\end{align}
with\footnote{we recall that $a_l$ is the limit of $a_i$ when $i$ tends to $+\infty$} $M=a_l+\lambda$.
We denote $T_0$ the first time-stamp such that $r_{N(T_0)-1}+\sqrt{log(T)/(N(T_0)-1)} \leq r_l$ or $r_{N(T_0)-1}-\sqrt{log(T)/(N(T_0)-1)} > r_l$. Therefore, adopting the algorithm \ref{euclid} referred by $\pi$, we have from $0$ till $T_0-1$, the threshold applied is $0$.

Now we find a lower bound of $N(t)$ with high probability for a given $t$ in $[0,T_0-1]$.

We let $c_i$ denotes the time elapsed from the first time that we start transmitting after $t_{i}$ till the first time the channel is good or the transmission is successful, then, $c_i$ follows a geometric distribution with parameter $\rho$. Accordingly, the probability that $c_i=k$ is $(1-\rho)^k \rho$. As consequence, the expectation of $c_i$ is $\frac{1-\rho}{\rho}=c$.

\begin{proposition}\label{prop:probability_N_t_threshold_0}
When $t \in [0,T_0-1]$, we have that:
$$P(N(t) > \frac{1}{c^2}(\sqrt{tc}(\sqrt{tc}-\sqrt{log(T)})-1) \geq 1-\frac{2}{T^2}$$
\end{proposition}
\begin{proof}
See appendix \ref{app:prop:probability_N_t_threshold_0}.
\end{proof}

To that extent, we consider the event $H(t)=\{N(t) \geq \frac{1}{c^2)}(\sqrt{tc}(\sqrt{tc}-\sqrt{log(T)})-1\}$.
\begin{proposition}\label{prop:T_0_less_than_L_0}
Knowing $G(T)$ and $H(t)$ for $t \in [0,T_0-1]$.
Denoting $\lfloor log(T)/c \Big (1+c\sqrt{\frac{4}{|r^*-r_l|^2}+2} \Big )^2 \rfloor$\footnote{$\lfloor x \rfloor$ is the integral part of $x$} $+1$ by $L_0$, then $T_0 \leq L_0$.  
\end{proposition}
\begin{proof}
See appendix \ref{app:prop:T_0_less_than_L_0}.
\end{proof}
Then, we have:
\begin{align}\label{regret_ineq_2}
&E(\sum_{t=0}^T |C(a^{\pi}(t),d^{\pi}(t))-C^*||G(T)) \nonumber\\
=&E(\sum_{t=0}^{T_0-1}|C(a^{\pi}(t),d^{\pi}(t))-C^*||G(T))\nonumber\\
&+E(\sum_{t=T_0}^T |C(a^{\pi}(t),d^{\pi}(t))-C^*||G(T))
\end{align} 
\begin{remark}
For t$ \in [0,T_0-1]$, since $H(t)$ and $G(T)$ are independent, then $P(H(t)|G(T))=P(H(t))$
\end{remark}
We focus on the first term:
\begin{align}\label{regret_ineq_3}
&E(\sum_{t=0}^{T_0-1} |C(a^{\pi}(t),d^{\pi}(t))-C^*||G(T)) \nonumber \\
&=E(\sum_{t=0}^{T_0-1} |C(a^{\pi}(t),d^{\pi}(t))-C^*||G(T),H(t))\nonumber \\
& \ \ \ \ \ \ \ \ \times P(H(t)|G(T))\nonumber \\
+&E(\sum_{t=0}^{T_0-1} |C(a^{\pi}(t),d^{\pi}(t))-C^*||G(T),\bar{H}(t))\nonumber \\
& \ \ \ \ \ \ \ \ \times P(\bar{H}(t)|G(T)) \nonumber \\
\leq& E(\sum_{t=0}^{L_0}|C(a^{\pi}(t),d^{\pi}(t))-C^*||G(T),H(t)) \nonumber \\
&+\frac{2M}{T}\nonumber \\
\leq& M L_0+ \frac{2M}{T}\nonumber \\ 
\end{align}

Now we deal with the second term. To that extent, we denote $A(T_0)$ the event $\{r_{N(T_0)-1}+\sqrt{log(T)/(N(T_0)-1)} \leq r_l\}$ and $B(T_0)$ the event $\{r_{N(T_0)-1}-\sqrt{log(T)/(N(T_0)-1)} > r_l\}$.
We have that:
\begin{align}\label{regret_ineq_4}
&E(\sum_{t=T_0}^T |C(a^{\pi}(t),d^{\pi}(t))-C^*||G(T)) \nonumber \\
=&E(\sum_{t=T_0}^T|C(a^{\pi}(t),d^{\pi}(t))-C^*||G(T),A(T_0))\nonumber \\
& \ \ \ \ \ \ \ \ \times P(A(T_0)|G(T))\nonumber \\
&+E(\sum_{t=T_0}^T|C(a^{\pi}(t),d^{\pi}(t))-C^*||G(T),B(T_0))\nonumber \\
& \ \ \ \ \ \ \ \ \times P(B(T_0)|G(T))
\end{align}
If $A(T_0)$ occurs\footnote{We consider in this case that $r^* < r_l$}, then, according to Algorithm \ref{euclid2}, we apply $0$ till the first time $t_i$ where $t_i$ satisfies $r_{i-1}+2\sqrt{\frac{log(T)}{i-1}}<r_l$. To that extent, we denote $T_1$ the first time after $T_0$, such that we have $r_{N(T_1)-1}+2\sqrt{\frac{log(T)}{N(T_1)-1}}<r_l$.

\begin{proposition}
Knowing $G(T)$ and $H(t)$ for $t \in [0,T_1-1]$. By letting $L_1$ be $\lfloor log(T)/c \Big (1+c\sqrt{\frac{9}{|r^*-r_l|^2}+2} \Big )^2 \rfloor +1$, then $T_1 \leq L_1$.  
\end{proposition}

\begin{proof}
We omit the proof as it follows the same procedure as done for Proposition \ref{prop:T_0_less_than_L_0}.
\end{proof}

Leveraging the above Proposition, we have that:
\begin{align}\label{regret_ineq_5}
&E(\sum_{t=T_0}^T |C(a^{\pi}(t),d^{\pi}(t))-C^*||G(T),A(T_0)) \nonumber \\
=&E(\sum_{t=T_0}^{T_1-1} |C(a^{\pi}(t),d^{\pi}(t))-C^*||G(T),A(T_0))\nonumber \\
&+E(\sum_{t=T_1}^T|C(a^{\pi}(t),d^{\pi}(t))-C^*||G(T),A(T_0)) \nonumber \\ 
\leq& ML_1+E(\sum_{t=T_1}^T |C(a^{\pi}(t),d^{\pi}(t))-C^*||G(T),A(T_0))
\end{align}
Knowing $G(T)$ and $A(T_0)$, we have for $t \geq T_1$, $r_{N(t)}$ is less than $r^*+\sqrt{log(T)/N(L_1)} < r_l$. Hence, by letting $a$ be $r^*+\sqrt{log(T)/N(L_1)}$, for all $r \leq a$, the function $r \rightarrow n(r)$ that represents the optimal threshold under the transition probability $r$ of the source, is lipchitz function. To that extent, we denote by $C_1 \geq 1$ the constant such that for all $r,r' \in [0,a]^2$, if $|r-r'|< \delta$, then $|n(r)-n(r')|< C_1\delta$.
Moreover the function $n(r)$ is upper bounded by a constant denoted by $m$ for $0 \leq r \leq a$.  
\begin{proposition}\label{prop:probability_N_t_threshold_n_r}
Denoting by $t'$, $t-T_1$, and by $N'(t)$, $N(t)-N(T_1)$, then we have for $t\geq T_1$:
\begin{align}
&P(N'(t) > \frac{1}{(m+c)^2}(\sqrt{t'(m+c)}(\sqrt{t'(m+c)}-\sqrt{log(T)})-1 \nonumber \\
& \ \ \ \ \ \ \ \ \ \ \ \ \ |G(T),A(T_0))\geq 1-\frac{2}{T^2}
\end{align}
\end{proposition}
\begin{proof}
See appendix \ref{app:prop:probability_N_t_threshold_n_r}
\end{proof}
We denote $W(t)$ the event $\{N'(t) > \frac{1}{(m+c)^2}(\sqrt{t'(m+c)}(\sqrt{t'(m+c)}-\sqrt{log(T)})-1\}$.
Accordingly, we have that:
\begin{align}\label{regret_ineq_6}
&E(\sum_{t=T_1}^T |C(a^{\pi}(t),d^{\pi}(t))-C^*||G(T),A(T_0))\nonumber \\
=&E(\sum_{t=T_1}^T |C(a^{\pi}(t),d^{\pi}(t))-C^*||G(T),A(T_0),W(t))\nonumber \\
& \ \ \ \ \ \ \ \times P(W(t)|G(T),A(T_0))\nonumber \\
+&E(\sum_{t=T_1}^T |C(a^{\pi}(t),d^{\pi}(t))-C^*||G(T),A(T_0),\bar{W}(t))\nonumber \\
& \ \ \ \ \ \ \ \times P(\bar{W}(t)|G(T),A(T_0)) \nonumber \\
\leq& E(\sum_{t=T_1}^T|C(a^{\pi}(t),d^{\pi}(t))-C^*||G(T),A(T_0),W(t))\nonumber \\
&+\frac{2M}{T}
\end{align}
Leveraging the above equations, our goal in the sequel will be to bound the term $E(\sum_{t=T_1}^T|C(a^{\pi}(t),d^{\pi}(t))-C^*||G(T),A(T_0),W(t))$. 
To that extent we provide this following proposition.
We have:
\begin{proposition}\label{prop:similarity_prob_pi_and_optim_thresh}
There exists a constant $C_2$ such that for $L_2=C_2log(T)$, we have that:\\
\begin{align}
&P(F(L_2))\nonumber \\
&=P(\{ \underset{t\in[L_2,T]}{\cup} \pi(s^{\pi}(t)) \neq \pi^*(s^{\pi}(t))|G(T), A(T_0),W(t)\})\nonumber \\
&\leq \frac{2}{T}
\end{align}
\end{proposition}

\begin{proof}
See Appendix \ref{app:prop:similarity_prob_pi_and_optim_thresh}
\end{proof}

We have that:
\begin{align}
&E(\sum_{t=T_1}^T|C(a^{\pi}(t),d^{\pi}(t))-C^*||G(T),A(T_0),W(t))\nonumber \\
=&E(\sum_{t=T_1}^{L_2-1}|C(a^{\pi}(t),d^{\pi}(t))-C^*||G(T),A(T_0),W(t))\nonumber \\
&+E(\sum_{t=L_2}^T|C(a^{\pi}(t),d^{\pi}(t))-C^*||G(T),A(T_0),W(t))\nonumber \\
&\leq M L_2 \nonumber \\
&+E(\sum_{t=L_2}^T|C(a^{\pi}(t),d^{\pi}(t))-C^*||G(T),A(T_0),W(t))
\end{align}
We have for $t\geq L_2$, leveraging the above Proposition:
\begin{align}
&E(\sum_{t=L_2}^T|C(a^{\pi}(t),d^{\pi}(t))-C^*||G(T),A(T_0),W(t)) \nonumber \\
=&E(\sum_{t=L_2}^T|C(a^{\pi}(t),d^{\pi}(t))-C^*||G(T),A(T_0),W(t),F(T_1)) \nonumber \\
& \ \ \ \ \ \ \ \times P(F(L_2)) \nonumber \\
+&E(\sum_{t=L_2}^T|C(a^{\pi}(t),d^{\pi}(t))-C^*||G(T),A(T_0),W(t),\bar{F}(L_2)) \nonumber \\
& \ \ \ \ \ \ \ \times P(\bar{F}(L_2)) \nonumber \\
\leq & E(\sum_{t=T_1}^{L_2-1}|C(a^{\pi}(t),d^{\pi}(t))-C^*||G(T),A(T_0),W(t),F(L_2)) \nonumber \\
&\leq 2M \\
&+E(\sum_{t=L_2}^T|C(a^{\pi}(t),d^{\pi}(t))-C^*||G(T),A(T_0),W(t),\bar{F}(L_2))
\end{align}
In the sequel, we bound the term $E(\sum_{t=L_2}^T|C(a^{\pi}(t),d^{\pi}(t))-C^*||G(T),A(T_0),W(t),\bar{F}(L_2))$. Given $\bar{F}(L_2)$, then $\pi(s(t))=\pi^*(s(t))$ for all $t\geq L_2$. Hence starting from $L_2$, we apply exactly the optimal threshold $n^*$. To that extent, we denote by $u_i(t)$ the probability that MAoII is at state $i$ at time $t \geq L_2$ under threshold policy $n^*$. We have that:
\begin{align}
E(C(a^{\pi}(t),d^{\pi}(t))| A(T_0),\bar{F(L_2)})=\sum_{i=0}^{+\infty} u_i(t)a_i+\lambda \sum_{i=n^*}^{+\infty}u_i(t)
\end{align}
On the other hands:
\begin{align}
C^*=\sum_{i=0}^{+\infty} u_i^*a_i+\lambda\sum_{i=n^*}^{+\infty}u_i^*
\end{align}
Therefore:
\begin{align}
|C(a^{\pi}(t),d^{\pi}(t))-C^*|\leq& \sum_{i=0}^{+\infty} |u_i(t)-u_i^*|a_i \nonumber \\
&+\lambda|\sum_{i=n^*}^{+\infty}u_i(t)-u_i^*|
\end{align}
We need to find a bound for $|u_i(t)-u_i^*|$ that depends on $t$. To that end, we need to express the distribution of DTMC at time $t+1$ in function of the one at time $t$. In order to easily manipulate and analyze the evolution of $u(t)=(u_0(t),u_1(t),\cdots)^{\top}$, we restrict our analysis only to the $n^*$ first terms of the vector $u(t)$ ($n^*$ is finite knowing $A(T_0)$). While, for the states greater or equal to $n^*$, we consider that they constitute one state. Without loss of generality, we denote this state by $n^*$. In other words, we consider the vector $u'(t)=(u_0(t),u_1(t),\cdots,u'_{n^*}(t))^{\top}$ where $u'_{n^*}(t)=\sum_{i=n^*}^{+\infty}u_i(t)$. As consequence, the transition probability (the relation between $u'(t+1)$ and $u'(t)$) considering this new vector is equal:
\begin{equation}Q=\left[\begin{array}{ccccccc}
0&0&\cdots &\cdots &0 &\rho \\
1&0&\cdots &\cdots &0  &0 \\
0 &1 &\ddots & &  &0 \\
&\ddots&1 &\ddots&&\vdots\\
 &&\ddots&\ddots &0 &0\\
0&0&\cdots&0&1&1-\rho\\
\end{array}\right].
\end{equation}
Moreover as $\sum_{i=0}^{n^*}u^{'}_i(t)=1$, we replace $u^{'}_{n^*}(t)$ by $1-\sum_{i=0}^{n^*-1}u^{'}_i(t)$, and we get the following transition probability by omitting the element $u^{'}_{n^*}(t)$ from the vector $u^{'}(t)$: 
\begin{equation}Q=\left[\begin{array}{ccccccc}
-\rho&-\rho&\cdots &\cdots &-\rho &-\rho \\
1&0&\cdots &\cdots &0  &0 \\
0 &1 &\ddots & &  &0 \\
&\ddots&1 &\ddots&&\vdots\\
 &&\ddots&\ddots &0 &0\\
0&0&\cdots&0&1&0\\
\end{array}\right].
\end{equation}
Thus, $u'(t+1)=Qu'(t)+v$, where $u'(t)=(u_1(t),u_2(t), \cdots, u_{n^*-1}(t))^{\top}$ and $v=(\rho,0,\cdots,0)^{\top}$.
On the other hands, by definition of the stationary distribution, $u^*=Qu^*+v$, where  $u^*(t)=(u_1^*,u_2^*, \cdots, u_{n^*-1}^*)^{\top}$
Hence, $u'(t+1)-u^*=Q(u'(t)-u^*)$.
Consequently: $||u'(t)-u^*|| \leq ||Q||^t ||u(L_2)-u^*||$.
Our aim will be then to prove that $||Q||$ is strictly less than $1$.
\begin{proposition}\label{prop:spectral_value_less_one}
The spectral value of $Q$ denoted by $\gamma$ is strictly less than 1.
\end{proposition}
\begin{proof}
See Appendix \ref{app:prop:spectral_value_less_one}.
\end{proof}
Leveraging this proposition, we have that $\sum_{i=0}^{n^*-1} |u_i(t)-u_i^*|a_i \leq \sum_{i=0}^{n^*-1} \gamma^t ||u_i(L_2)-u_i^*|| a_l$. And, $\lambda|\sum_{i=n^*}^{+\infty} u_i(t)-\sum_{i=n^*}^{+\infty} u_i^*|=\lambda|\sum_{i=0}^{n^*-1} u_i(t)-\sum_{i=0}^{n^*-1} u_i^*| \leq \sum_{i=0}^{n^*-1} \lambda\gamma^t ||u_i(L_2)-u_i^*||$.\\
We still have to deal with $\sum_{i=n^*}^{+\infty} |u_i(t)-u_i^*|a_i$.
We suppose that at time $t=L_2$, MAoII is at state\footnote{The analysis follows the same steps starting from a different state than $a_0$ at time $L_2$} $a_0$. Hence, we have:
\begin{align}
&\sum_{i=n^*}^{+\infty} |u_i(t)-u_i^*|a_i \nonumber \\
&=\sum_{i=n^*}^{t-L_2} |u_i(t)-u_i^*|a_i+\sum_{i=t-L_2+1}^{+\infty} |u_i(t)-u_i^*|a_i \nonumber \\  
&\leq  \sum_{i=n^*}^{t-L_2} |(1-\rho)^{i-n^*}u_0(t-i)-(1-\rho)^{i-n^*}u_0^*|a_i  \nonumber \\
&+ \sum_{i=t-L_2+1}^{+\infty} (1-\rho)^{i-n^*}u_0^*a_i \nonumber\\
&\leq^{(a)}   \sum_{i=n^*}^{t-L_2}a_l (1-\rho)^{i-n^*}\gamma^{t-i}||u(L_2)-u^*|| \nonumber \\
&+  \sum_{i=t-L_2+1}^{+\infty} a_l (1-\rho)^{i-n^*}u_0^* \nonumber \\
&\leq   a_l \frac{\gamma^{t-L_2-n^*+1}-(1-\rho)^{t-L_2-n^*+1}}{\gamma-(1-\rho)}||u(L_2)-u^*|| \nonumber \\
&+ a_l (1-\rho)^{t-L_2-n^*+1} \frac{1}{\rho} u_0^* 
\end{align}
$(a)$ comes from the fact that at time $t$, for all $i > t-L_2$, $u_i(t)=0$ since we move at most by one state at each time slot.  
Therefore: 
\begin{align}
&|C(a^{\pi}(t),d^{\pi}(t))-C^*| \nonumber \\
&\leq \sum_{i=0}^{n^*-1} \gamma^{t-L_2} ||u_i(L_2)-u_i^*|| a_l \nonumber \\
&+a_l \frac{\gamma^{t-L_2-n^*+1}-(1-\rho)^{t-L_2-n^*+1}}{\gamma-(1-\rho)}||u(L_2)-u^*|| \nonumber \\
&+ a_l (1-\rho)^{t-L_2-n^*+1} \frac{1}{\rho} u_0^*+\sum_{i=0}^{n^*-1} \lambda\gamma^t ||u_i(L_2)-u_i^*||
\end{align}
As consequence:
\begin{align}\label{regret_ineq_7}
&E(\sum_{t=L_2}^T|C(a^{\pi}(t),d^{\pi}(t))-C^*||G(T),A(T_0),W(t),\bar{F}(L_2)) \nonumber \\ 
&\leq \sum_{t=0}^T \Big [ \sum_{i=0}^{n^*-1} \gamma^t ||u_i(L_2)-u_i^*|| a_l \nonumber \\
&+ a_l \frac{\gamma^{t-n^*+1}-(1-\rho)^{t-n^*+1}}{\gamma-(1-\rho)}||u(L_2)-u^*|| \nonumber \\
&+ a_l (1-\rho)^{t-n^*+1} \frac{1}{\rho} u_0^*+\sum_{i=0}^{n^*-1} \lambda\gamma^t ||u(L_2)-u^*|| \Big ] \nonumber \\
&\leq K_1
\end{align}
      
where $K_1$ is a constant independent of $T$.   
If $B(T_0)$ occurs\footnote{we consider that $r^* \geq r_l$}, then according to Algorithm \ref{euclid2}, we apply the infinite threshold. Knowing the event $G(T)$, $r^* \geq r_l$. That means the optimal threshold will be effectively infinite knowing $G(T)$. Thus, $C^*=a_l$ and $C(a^{\pi}(t),d^{\pi}(t))=a_{t-T_0}$. 
\begin{align}\label{regret_ineq_8}
&E(\sum_{t=T_0}^T|C(a^{\pi}(t),d^{\pi}(t))-C^*||G(T),B(T_0))\nonumber\\ &=\sum_{t=T_0}^T |a_{t-T_0}-a_l| \nonumber \\
& \leq \sum_{t=T_0}^T \frac{(p-r)^{t-T_0+1}}{Nr}+\frac{(1-r)^{t-T_0+1}}{r} \nonumber \\
&\leq K_2
\end{align}
where $K_2$ a constant independent of $T$.
Therefore, combining \eqref{regret_ineq_1},\eqref{regret_ineq_2},\eqref{regret_ineq_3},\eqref{regret_ineq_4},\eqref{regret_ineq_5},\eqref{regret_ineq_6}, \eqref{regret_ineq_7} and \eqref{regret_ineq_8}, we get our desired result, i.e., there exists a constant $K$ such that $R_{\pi}(T)\leq K log(T)$. Hence, we proved Theorem \ref{theo:bound_regret}.

\section{Numerical Results}\label{sec:num_reslt}
Our goal in this section is to evaluate the performance of our proposed solution given in Algorithm \ref{euclid2} and compare it with the greedy policy. This later consists of applying at each $t_i$, the optimal threshold with respect to the estimated transition probability at $t_i$ which is $r_{i-1}$ in the function of $T$.\\
To that extent, we showcase the evolution of the regret function under our proposed solution when the true value of $r^*$ is greater than $r_l$. Specifically, we consider these following settings:
\begin{itemize}
\item $r^*=0.25$.
\item $\rho_1=0.5$.
\item $N=5$.
\item $\lambda=8$.
\item $r_l=0.2212$\footnote{By definition of $r_l$, we get $r_l$ by resolving the equation $\frac{(N-1)(N+1-N r_l) \rho}{N^2 r_l^2}=\lambda$}.
\end{itemize}
Then we compare the regret function under our proposed solution with the one under the greedy policy when the true value of $r^*$ is strictly less than $r_l$, i.e., the optimal threshold is finite. 
We show numerically that our proposed algorithm outperforms the greedy one when the real value of $r^*$ is strictly less than $r_l$.  
To that extent, we consider the respective parameters:
\begin{itemize}
\item $r^*=0.1$.
\item $\rho_1=0.5$.
\item $N=5$.
\item $\lambda=8$.
\item $r_l=0.2212$
\end{itemize} 

\begin{figure} 
\centering
\includegraphics[scale=0.6]{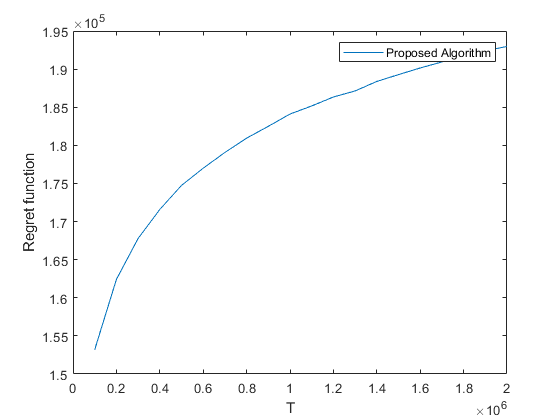}
\caption{Evolution of the regret function under the proposed policy}
\label{fig:comp_gred_propos_alg_r_greater_rlim}
\end{figure}
\begin{figure}
\centering
\includegraphics[scale=0.6]{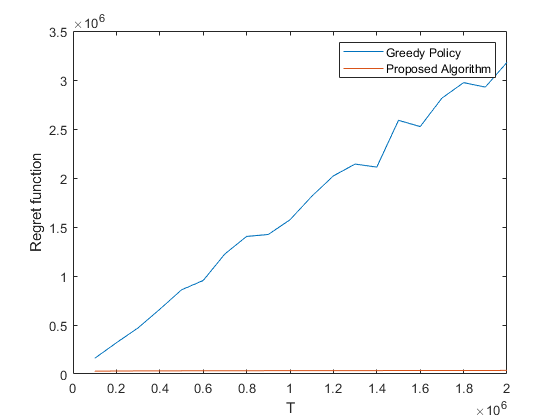}
\caption{Comparison between the greedy policy and the proposed algorithm in terms of the regret function when $r^*$ is less than $r_l$}
\label{fig:comp_gred_propos_alg_r_less_rlim}
\end{figure}
In Figure \ref{fig:comp_gred_propos_alg_r_greater_rlim}, one can notice that our solution gives us a logarithmic regret when $T$ grows.

In Figure \ref{fig:comp_gred_propos_alg_r_less_rlim}, one can observe that our proposed algorithm gives us a sub-linear regret function, precisely a logarithmic regret. Whereas the greedy policy gives us a linear regret. This is because the greedy policy always applies the optimal threshold corresponding to the estimated parameter, which can definitively stop the exploration when encountering an estimated value greater than $r_l$ even in the earliest steps of the exploration process.   
Consequently, our algorithm outperforms the greedy policy in minimizing the regret function.


In conclusion, our developed algorithm turns out to be essential to ensure a logarithmic regret whatever the value of the real $r^*$ since the greedy policy fails to reach this goal when $r^* < r_l$.


\section{Conclusion}\label{sec:concl}
In this paper, we considered the problem of remote monitoring of an unknown source where a central entity decide whether to schedule the source or not in order to receive the new updates under energy constraint. We established that the optimal policy is a threshold based policy. When the source parameters are known, we have provided a simple algorithm that finds the optimal threshold policy. When the source parameters are unknown, we developed an online reinforcement learning algorithm that gives a good balance between the exploration-exploitation trade-off. We proved that the regret function under our proposed algorithm is less than $O(log(T))$.
Finally, we have provided numerical results that highlight the performance of our proposed policy compared to the greedy one.

\bibliographystyle{IEEEtran} 

\bibliography{bibliography}

\begin{appendices}

\section{Proof of Theorem \ref{theo:threshold_policy_discounted}}\label{app:theo:threshold_policy_discounted}
In this proof, we distinguish between two cases:
\begin{itemize}
\item $1-r \geq |p-r|$
\item $1-r < |p-r|$
\end{itemize}
We start by the first case:
\begin{itemize}
\item $1-r \geq |p-r|$:\\
We provide first an useful lemma. 
\begin{lemma}\label{lem:MAoII_increasing}
$a_{j}$ is increasing with $j$
\end{lemma} 

\begin{IEEEproof} 
The explicit expression of $a_{j}$ is:
\begin{equation}\label{eq:expression_maoii}
a_{j}=\frac{N-1}{Nr}+\frac{(p-r)^{i+1}}{Nr}-\frac{(1-r)^{i+1}}{r}
\end{equation}
Therefore, after some computations and mathematical analysis, we obtain:
\begin{equation}
a_{j+1}-a_{j}=(1-r)^{j+1}-(p-r)^{j+1}
\end{equation}
Given that $0 \leq |p-r| \leq 1-r$, then $(p-r)^{j+1} \leq |p-r|^{j+1} \leq (1-r)^{j+1}$. Therefore, $(1-r)^{j+1}-(p-r)^{j+1} \geq 0$.  
Hence, $a_{j}$ is increasing with $j$.
\end{IEEEproof}

Based on this lemma, we prove the following lemma.
\begin{lemma}\label{lem:V_increasing}
$V_{\beta}(.)$ is increasing with $a_j$.
\end{lemma}
\begin{IEEEproof}
We prove the present lemma by induction using the Value iteration equation \eqref{eq:bellman_equation_time_t}. In fact, we show that $V^t_{\beta}(\cdot)$ is increasing and we conclude for $V_{\beta}(\cdot)$.\\
As $V^0_{\beta}(.)=0$, then the property holds for $t=0$.
If $V^t_{\beta}(.)$ is increasing with $a$, we show that for $a_{j} \leq a_{i}$, $V^t_{\beta,0}(a_{j}) \leq V^t_{\beta,0}(a_{i})$ and $V^t_{\beta,1}(a_{j}) \leq V^t_{\beta,1}(a_{i})$ where for each $k \in \mathbf{N}$:
\begin{align}
V^t_{\beta,0}(a_{k})&=a_{k}+\beta V^t_{\beta}(a_{k+1}) \\
V^t_{\beta,1}(a_{k})&=a_{k}+\lambda+\rho \beta V^t_{\beta}(a_{0})+(1-\rho)\beta V^t_{\beta}(a_{k+1}) 
\end{align}
We have that:
\begin{equation}
V^{t+1}_{\beta,0}(a_{j}) - V^{t+1}_{\beta,0}(a_{i})=a_{j}-a_{i}+\beta (V^t_{\beta}(a_{j+1}) - V^t_{\beta}(a_{i+1}))
\end{equation}
According to Lemma \ref{lem:MAoII_increasing}, given that $a_{j} \leq a_{i}$, then $j \leq i$. That means $a_{j+1} \leq a_{i+1}$. Therefore, since $V^t_{\beta}(.)$ is increasing with $a_j$, we have that:
$$V^{t+1}_{\beta,0}(a_{j+1}) - V^{t+1}_{\beta,0}(a_{i+1}) \leq 0$$
As consequence, $V^{t+1}_{\beta,0}(\cdot)$ is increasing with $a_{j}$.\\
In the same way, we have:
$$V^{t+1}_{\beta,1}(a_{j}) - V^{t+1}_{\beta,1}(a_{i})=a_{j}-a_{i}+(1-\rho)\beta (V^t_{\beta}(a_{j+1}) - V^t_{\beta}(a_{i+1}))$$
Hence: 
\begin{equation}
V^{t+1}_{\beta,1}(a_{j}) - V^{t+1}_{\beta,1}(a_{i}) \leq 0
\end{equation}
As consequence, $V^{t+1}_{\beta,1}(\cdot)$ is increasing with $a_{j}$.\\
Since $V^{t+1}_{\beta}(.)=\min\{V^{t+1}_{\beta,0}(\cdot),V^{t+1}_{\beta,1}(\cdot)\}$, then $V^{t+1}_{\beta}(.)$ is increasing with $a_j$. Accordingly, we demonstrate by induction that $V^t_{\beta}(.)$ is increasing for all $t$. Knowing that $\underset{t \rightarrow +\infty}{\text{lim}} V^t_{\beta}(a_j)=V_{\beta}(a_j)$, $V_{\beta}(.)$ must be also increasing with $a_j$.
\end{IEEEproof}

We define:
\begin{equation}
\Delta V_{\beta}(a_{j})=V_{\beta,1}(a_j)-V_{\beta,0}(a_j)
\end{equation}
where $\underset{t \rightarrow +\infty}{\text{lim}} V_{\beta,0}^t(a_j)=V_{\beta,0}(a_j)$ and $\underset{t \rightarrow +\infty}{\text{lim}} V_{\beta,1}^t(a_j)=V_{\beta,1}(a_j)$.\\
Subsequently, $\Delta V_{\beta}(a_{j})$ equals to:
\begin{equation}
\Delta V_{\beta}(a_{j})=\rho \beta [\frac{\lambda}{\rho \beta }+V_{\beta}(a_0)-V_{\beta}(a_{j+1})]
\end{equation}
According to Lemma \ref{lem:V_increasing}, $V_{\beta}(.)$ is increasing with $a_{j+1}$. Therefore, $\Delta V_{\beta}(a_{j})$ is decreasing with $a_j$. Hence, there exists $a_{n}$ such that for all $a_j \leq a_{n}$, $\Delta V_{\beta}(a_{j})\geq 0$, and for all $a_j > a_{n}$, $\Delta V_{\beta}(a_{j}) <0$. Given that the optimal action for state $a_j$ is the one that minimizes $\min\{V_{\beta,0}(\cdot),V_{\beta,1}(\cdot)\}$, then for all $a_j \leq a_{n}$, the optimal decision is to stay idle since $\min\{V_{\beta,0}(a_j),V_{\beta,1}(a_j)\}=V_{\beta,0}(a_j)$, and for all $a_j > a_{n}$, the optimal decision is to transmit since $\min\{V_{\beta,0}(a_j),V_{\beta,1}(a_j)\}=V_{\beta,1}(a_j)$. Specifically, as $a_{j}$ is increasing with $j$, there exists $n$ such that for all $j < n$, the optimal action is the passive action, and for all $j \geq n$, the optimal action is the active one. Hence, we establish that when $1-r \geq |p-r|$, the optimal threshold is a threshold based policy.
We move now to the second case that require more analysis.

\item $1-r < |p-r|$:\\
Our approach for this case will be first to give the structural form of MAoII function $a_j$. Then, based on that, we prove that the optimal solution is threshold policy. 
Indeed, unlike the first case, we show that $a_j$ is an oscillating function with $j$. In order to establish that, we proceed first by giving remark and lemmas as follows:
\begin{remark}\label{rem:two_states}
Given that $p+(N-1)r=1$, then $(N-1)r \leq 1$. That means $r \leq \frac{1}{N-1}$. 
If $1-r < |p-r|$, then $p$ is necessarily less than $r$. Which means that $1-r < r-p$. Thus, $r > \frac{1}{2}$.
Leveraging these two results, $\frac{1}{N-1} > \frac{1}{2}$. Hence, $N < 3$. Thereby, $N=2$.($N>1$)  
\end{remark}
\begin{lemma}\label{lem:growth_age}
$a_{2k}$ is increasing with $k$ and under Assumption \ref{assump:volatility}, $a_{2k+1}$ is decreasing with $k$.
\end{lemma}
\begin{proof}
We have:
\begin{align}
a_{i+2}-a_i= (1-r)^{i+1}(2-r)-(p-r)^{i+1}(1+p-r)
\end{align}
If $i=2k$, since $p-r<0$, then $(p-r)^{2k+1}(1+p-r)<0$. Therefore $a_{2(k+1)}-a_{2k} \geq 0$.\\
If $i=2k+1$, we have: 
\begin{align}
a_{2k+1+2}-a_{2k+1}=&[(1-r)^{2k+2}(2-r) \nonumber\\
&-(p-r)^{2k+2}(1+p-r)] \nonumber\\
=&[(1-r)^{2k}(1-r)^{2}(2-r) \nonumber\\
&-(p-r)^{2k}(p-r)^{2}(1+p-r)]
\end{align}  
We have $(1-r)^{2k} \leq (p-r)^{2k}$. 
In the sequel, we prove that if $r>\frac{5}{7}$, then $a_{2k+1}$ is decreasing with $k$.
For that we investigate the sign of $(1-r)^{2}(2-r)- (p-r)^{2}(1+p-r)$ in function of $r$ to establish our desired result.
To that end, we replace $p$ by $1-(N-1)r$ and we get the following inequality that we should prove:
\begin{equation}
(1-r)^2(2-r) \leq (1-Nr)^2(2-Nr)
\end{equation}
The difference $(1-r)^2(2-r) \leq (1-Nr)^2(2-Nr)$ is equal to $5r(N-1)-4r^2(N^2-1)+r^3(N^3-1)$. As was mentioned in Remark \ref{rem:two_states}, $N=2$. Accordingly, we should prove that $5r-12r^2+7r^3 \leq 0$. As $r >0$, then we must demonstrate that the second degree polynomial $5-12r+7r^2$ is negative when $r \geq \frac{5}{7}$.
\begin{lemma}\label{lem:sign_polynom_1}
If $r \geq \frac{5}{7}$, then $5-12r+7r^2\leq 0$
\end{lemma}
\begin{proof}
Resolving the equation $5-12r+7r^2=0$, we get two different roots which are $\frac{5}{7}$ and $1$. That means in $[\frac{5}{7},1]$, $5-12r+7r^2$ is less than $0$.
Hence, $5-12r+7r^2\leq 0$ for $r \geq \frac{5}{7}$. That completes the proof.    
\end{proof}
According to Assumption \ref{assump:volatility}, $(N-1)r=r \geq 4p$. Therefore, $r=1-p \geq 1-\frac{r}{4}$. Thus, $1>r\geq \frac{4}{5}\geq \frac{5}{7}$. Then according to Lemma \ref{lem:sign_polynom_1}, $5-12r+7r^2\leq 0$. 
Leveraging that and the fact that $(1-r)^{2k} \leq (p-r)^{2k}$, then $a_{2(k+1)+1}-a_{2k+1} \leq 0$. Hence $a_{2k+1}$ is decreasing in $k$. 
\end{proof}

\begin{lemma}
For all $(k,k') \in [0,\mathbb{N}]^2$, $a_{2k} \leq a_{2k'+1}$. 
\end{lemma}
\begin{proof}
Baring in mind the equation (\ref{eq:expression_maoii}), we have $a_k$ converges to$\frac{N-1}{Nr}$. The same applies to the sub-sequences $a_{2k}$ and $a_{2k+1}$. Given that $a_{2k}$ is increasing, then $a_{2k} \leq \frac{N-1}{Nr}$ for all $k$. Similarly, as $a_{2k+1}$ is decreasing when $k \geq 0$, then $a_{2k+1} \geq \frac{N-1}{Nr}$. We deduce that for all $(k,k') \in [0,\mathbb{N}]^2$, $a_{2k} \leq a_{2k'+1}$.
\end{proof}
As these lemmas above have been laid out, we are now able to prove the Theorem \ref{theo:threshold_policy_discounted}. For that purpose, our main challenge is to show that $V_{\beta}^t(a_i)$ is increasing with $a_i$ since this result require intricate and non trivial mathematical analysis to prove it. 
For that, based on lemmas above, we show by induction these following statements for all $t$:
\begin{itemize}
\item (a) $V_{\beta}^t(a_i)$ is increasing with $a_i$. In other words, if $a_i \leq a_j$, then $V_{\beta}^t(a_i) \leq V_{\beta}^t(a_j)$
\item (b) $a_{2k+2}-a_{2k} \geq (V_{\beta}^t(a_{2k+1})-V_{\beta}^t(a_{2k+3}))$
\item (c) $a_{2k+1}-a_{2k+3} \geq (V_{\beta}^t(a_{2k+4})-V_{\beta}^t(a_{2k+2}))$
\item (d) $a_{2k+1}-a_{2k} \geq (V_{\beta}^t(a_{2k+1})-V_{\beta}^t(a_{2k+2}))$
\end{itemize}
These statements hold for $t=0$ since $V^0_{\beta}(a_k)=0$ for all $k \geq 0$.\\
We consider that these statements hold for a given $t$. Then, we prove that this is also the case for $t+1$.\\
We start first by showing the first point, i.e., $V_{\beta}^{t+1}(a_k) \leq V_{\beta}^{t+1}(a_{k+1})$.
\begin{itemize}

\item The growth of the function $V_{\beta}^{t+1}(.)$:\\
To proceed so, we show successively these tree following points:
\begin{enumerate}
\item $V_{\beta}^{t+1}(a_{2k}) \leq V_{\beta}^{t+1}(a_{2k+2})$
\item $V_{\beta}^{t+1}(a_{2k+1}) \leq V_{\beta}^{t+1}(a_{2k+3})$
\item $V_{\beta}^{t+1}(a_{2k}) \leq V_{\beta}^{t+1}(a_{2k+1})$
\end{enumerate}
We start first by demonstrating the first point.
For that, we prove that $V_{\beta,0}^{t+1}(a_{2k}) \leq V_{\beta,0}^{t+1}(a_{2k+2})$ and $V_{\beta,1}^{t+1}(a_{2k}) \leq V_{\beta,1}^{t+1}(a_{2k+2})$:
We have 
\begin{align}
&V_{\beta,0}^{t+1}(a_{2k+2})-V_{\beta,0}^{t+1}(a_{2k})\\
&=a_{2k+2}-a_{2k}- \beta(V_{\beta}^t(a_{2k+1})-V_{\beta}^t(a_{2k+3}))
\end{align}
By induction assumption, $V_{\beta}^t(a_{2k+1})-V_{\beta}^t(a_{2k+3}) \leq a_{2k+2}-a_{2k}$. Therefore, $V_{\beta,0}^{t+1}(a_{2k+2})-V_{\beta,0}^{t+1}(a_{2k}) \geq 0$ since $a_{2k+2}-a_{2k}\geq 0$ according to Lemma \ref{lem:growth_age}. 
Following the same steps, we prove also that $V_{\beta,1}^{t+1}(a_{2k}) \leq V_{\beta,1}^{t+1}(a_{2k+2})$. 
Hence, we deduce that:
\begin{equation}\label{eq:first_inequality_for_growth_function}
V_{\beta}^{t+1}(a_{2k}) \leq V_{\beta}^{t+1}(a_{2k+2})
\end{equation} 
The second point is to establish that $V_{\beta,0}^{t+1}(a_{2k+1}) \geq V_{\beta,0}^{t+1}(a_{2k+3})$, and $V_{\beta,1}^{t+1}(a_{2k+1}) \geq V_{\beta,1}^{t+1}(a_{2k+3})$. Similarly to the first case, we have that:
\begin{align}
&V_{\beta,d}^{t+1}(a_{2k+1})-V_{\beta,d}^{t+1}(a_{2k+3})\nonumber\\
&=a_{2k+1}-a_{2k+3}- \beta(1-\mathbf{1}_{\{d=1\}}\rho)(V_{\beta}^t(a_{2k+4})-V_{\beta}^t(a_{2k+2}))
\end{align}
where $d \in \{0,1\}$.
Likewise, by induction assumption $V_{\beta}^t(a_{2k+4})-V_{\beta}^t(a_{2k+2}) \leq a_{2k+1}-a_{2k+3}$. Therefore, $V_{\beta,d}^{t+1}(a_{2k+1})-V_{\beta,d}^{t+1}(a_{2k+3}) \geq 0$. Hence:
\begin{equation}\label{eq:second_inequality_for_growth_function}
V_{\beta}^{t+1}(a_{2k+1}) \geq V_{\beta}^{t+1}(a_{2k+3})
\end{equation}
As for the last point, we prove that $V_{\beta}^{t+1}(a_{2k}) \leq V_{\beta}^{t+1}(a_{2k+1})$:
\begin{align}
&V_{\beta,d}^{t+1}(a_{2k+1})-V_{\beta,d}^{t+1}(a_{2k})\nonumber\\
&=a_{2k+1}-a_{2k}- \beta(1-\mathbf{1}_{\{d=1\}}\rho)(V_{\beta}^{t}(a_{2k+1})-V_{\beta}^{t}(a_{2k+2}))
\end{align} 
Likewise, by induction assumption, $V_{\beta}^{t}(a_{2k+1})-V_{\beta}^{t}(a_{2k+2}) \leq a_{2k+1}-a_{2k}$. Therefore, $V_{\beta,d}^{t+1}(a_{2k+1})-V_{\beta,d}^{t+1}(a_{2k}) \geq 0$ since $a_{2k+1}-a_{2k}=[(1-r)^{2k+1}-(p-r)^{2k+1}]\geq 0$. Therefore:
\begin{equation}\label{eq:third_inequality_for_growth_function} 
V_{\beta}^{t+1}(a_{2k+1}) \geq V_{\beta}^{t+1}(a_{2k})
\end{equation}
That concludes the result.\\
We recall that our goal through these analysis above is to show that $V_{\beta}^{t+1}(.)$ is an increasing function with $a_i$.
Whereas, we only proved that $V_{\beta}^{t+1}$ is increasing in the set $\{a_{2k}\}_{k \in \mathbf{N}}$ and increasing in the set $\{a_{2k+1}\}_{k \in \mathbf{N}}$. In other words, if we take $a_{2k} \leq a_{2k'}$, then $V_{\beta}^{t+1}(a_{2k}) \leq V_{\beta}^{t+1}(a_{2k'})$, and if we take $a_{2k+1} \leq a_{2k'+1}$, $V_{\beta}^{t+1}(a_{2k+1}) \geq V_{\beta}^{t+1}(a_{2k'+1})$.
We still have to prove that the growth property is valid in $\{a_{2k}\}_{k \in \mathbf{N}} \cup \{a_{2k+1}\}_{k \in \mathbf{N}}$.
Given that $a_{2k} \leq a_{2k'+1}$ for all $k$ and $k' \geq 0$, we just need to prove that $V_{\beta}^{t+1}(a_{2k}) \leq V_{\beta}^{t+1}(a_{2k'+1})$ for all $k \geq 0$ and $k'\geq 0$ to establish our desired result.
Indeed, if $k \leq k'$, that means according to Equations \eqref{eq:first_inequality_for_growth_function} and \eqref{eq:third_inequality_for_growth_function}, $V_{\beta}^{t+1}(a_{2k}) \leq V_{\beta}^{t+1}(a_{2k'}) \leq V_{\beta}^{t+1}(a_{2k'+1})$.
If $k \geq k'$, then according to Equations \eqref{eq:second_inequality_for_growth_function} and \eqref{eq:third_inequality_for_growth_function}, $V_{\beta}^{t+1}(a_{2k}) \leq V_{\beta}^{t+1}(a_{2k+1}) \leq V_{\beta}^{t+1}(a_{2k'+1})$. 
That concludes the proof.
Hence, the first point regarding the growth of $V_{\beta}^{t+1}(.)$ is established, i.e., for all $a_i\leq a_j$:
\begin{equation}\label{eq:growth_V_t+1}
V_{\beta}^{t+1}(a_i) \leq V_{\beta}^{t+1}(a_j)
\end{equation} 
\item 
Now we prove the next following points successively:
\begin{itemize}
\item $a_{2k+2}-a_{2k} \geq (V_{\beta}^{t+1}(a_{2k+1})-V_{\beta}^{t+1}(a_{2k+3}))$
\item $a_{2k+1}-a_{2k+3} \geq (V_{\beta}^{t+1}(a_{2k+4})-V_{\beta}^{t+1}(a_{2k+2}))$
\item $a_{2k+1}-a_{2k} \geq (V_{\beta}^{t+1}(a_{2k+1})-V_{\beta}^{t+1}(a_{2k+2}))$
\end{itemize}
We provide first this useful Lemma: 
\begin{lemma}\label{lem:growth_difference_by_two_steps_age}
For all integer $i$, we have that $|a_{i+2}-a_i|$ is decreasing with $i$
\end{lemma}
\begin{proof}
If $i=2k$, then:
\begin{align}
&|a_{i+3}-a_{i+1}|-|a_{i+2}-a_i|=(a_{i+1}-a_{i+3})-(a_{i+2}-a_i)\nonumber\\
&=-[(2-r)^2(1-r)^{i+1}-(p-r)^{i+1}(1+p-r)^2]\nonumber\\
&\leq 0
\end{align}
since $(p-r)^{i+1}\leq 0$.\\
If $i=2k+1$, then:
\begin{align}
&|a_{i+3}-a_{i+1}|-|a_{i+2}-a_i|=(a_{i+3}-a_{i+1})-(a_{i}-a_{i+2})\nonumber\\
&=[(2-r)^2(1-r)^{i+1}-(p-r)^{i+1}(1+p-r)^2]\nonumber\\
&=[(2-r)^2(1-r)^{2k+2}-(p-r)^{2k+2}(1+p-r)^2]\nonumber\\
& \leq 0
\end{align}
since $r \geq \frac{4}{5} \Longrightarrow (2-r)^2(1-r)^{2k+2}-(p-r)^{2k+2}(1+p-r)^2 \leq 0  \ \ \forall k \geq 0$ 
\end{proof}
As for the first point, we have that:
\begin{align}
&V_{\beta}^{t+1}(a_{2k+1})-V_{\beta}^{t+1}(a_{2k+3})\nonumber\\
& \leq \max\{(V_{\beta,0}^{t+1}(a_{2k+1})-V_{\beta,0}^{t+1}(a_{2k+3}))\nonumber\\
& \hspace{1.2 cm},(V_{\beta,1}^{t+1}(a_{2k+1})-V_{\beta,1}^{t+1}(a_{2k+3}))\}\nonumber\\
&=\underset{d=0,1}{\text{max}}\{a_{2k+1}-a_{2k+3}-\nonumber\\
&\hspace{1.5 cm} \beta(1-\mathbf{1}_{\{d=1\}}\rho)(V_{\beta}^{t}(a_{2k+4})-V_{\beta}^{t+1}(a_{2k+2}))\} 
\end{align}
By induction assumption, we have that $V_{\beta}^{t}(a_{2k+4})-V_{\beta}^{t}(a_{2k+2})\geq 0$. Thus, $V_{\beta}^{t+1}(a_{2k+1})-V_{\beta}^{t+1}(a_{2k+3}) \leq a_{2k+1}-a_{2k+3}$. 
According to Lemma \ref{lem:growth_difference_by_two_steps_age}, $|a_{2k+3}-a_{2k+1}|-|a_{2k+2}-a_{2k}| \leq 0$. As consequence, $a_{2k+1}-a_{2k+3} \leq a_{2k+2}-a_{2k}$. Then:
\begin{equation}\label{eq:first_inequality_V}
V_{\beta}^{t+1}(a_{2k+1})-V_{\beta}^{t+1}(a_{2k+3}) \leq a_{2k+2}-a_{2k}
\end{equation}
For the second point, we have that:
\begin{align}
&V_{\beta}^{t+1}(a_{2k+4})-V_{\beta}^{t+1}(a_{2k+2}) \nonumber\\
&\leq \max\{(V_{\beta,0}^{t+1}(a_{2k+4})-V_{\beta,0}^{t+1}(a_{2k+2}))\nonumber\\
&\hspace{1.5 cm} ,(V_{\beta,1}^{t+1}(a_{2k+4})-V_{\beta,1}^{t+1}(a_{2k+2}))\}\nonumber\\
&=\underset{d \in \{0,1\}}{\text{max}}\{a_{2k+4}-a_{2k+2}\nonumber\\
&\hspace{1.5cm }-(1-\mathbf{1}_{\{d=1\}}\rho)\beta(V_{\beta}^{t}(a_{2k+3})-V_{\beta}^{t+1}(a_{2k+5}))\} 
\end{align}
By induction assumption, we have that $V_{\beta}^{t}(a_{2k+3})-V_{\beta}^{t}(a_{2k+5})\geq 0$. Thus, $V_{\beta}^{t+1}(a_{2k+4})-V_{\beta}^{t+1}(a_{2k+2}) \leq a_{2k+4}-a_{2k+2}$. 
According to Lemma \ref{lem:growth_difference_by_two_steps_age}, $|a_{2k+4}-a_{2k+2}|-|a_{2k+3}-a_{2k+1}| \leq 0$. As consequence, $a_{2k+4}-a_{2k+2} \leq a_{2k+1}-a_{2k+3}$. Then:
\begin{equation}\label{eq:second_inequality_V}
V_{\beta}^{t+1}(a_{2k+4})-V_{\beta}^{t+1}(a_{2k+2}) \leq a_{2k+1}-a_{2k+3}
 \end{equation}   
We move now to the last point. 
\begin{align}
&V_{\beta}^{t+1}(a_{2k+1})-V_{\beta}^{t+1}(a_{2k+2}) \nonumber \\
&\leq \max\{(V_{\beta,0}^{t+1}(a_{2k+1})-V_{\beta,0}^{t+1}(a_{2k+2}))\nonumber\\
&\hspace{1.5 cm} ,(V_{\beta,1}^{t+1}(a_{2k+1})-V_{\beta,1}^{t+1}(a_{2k+2}))\}\nonumber\\
&=\underset{d \in \{0,1\}}{\text{max}}\{a_{2k+1}-a_{2k+2}\nonumber\\
&\hspace{1.5cm }-(1-\mathbf{1}_{\{d=1\}}\rho)\beta(V_{\beta}^{t}(a_{2k+3})-V_{\beta}^{t+1}(a_{2k+2}))\} \\
\end{align}
By induction assumption, we have that $V_{\beta}^{t}(a_{2k+3})-V_{\beta}^{t}(a_{2k+2})\geq 0$. Thus, $V_{\beta}^{t+1}(a_{2k+1})-V_{\beta}^{t+1}(a_{2k+2}) \leq a_{2k+1}-a_{2k+2}$. Given that $a_{2k+2} \geq a_{2k}$, then:
\begin{equation}\label{eq:third_inequality_V}
V_{\beta}^{t+1}(a_{2k+1})-V_{\beta}^{t+1}(a_{2k+2}) \leq a_{2k+1}-a_{2k}
\end{equation} 
\end{itemize}  
Combining Equations \eqref{eq:growth_V_t+1}, \eqref{eq:first_inequality_V}, \eqref{eq:second_inequality_V} and \eqref{eq:third_inequality_V}, we conclude that all the statements (a), (b), (c) and (d), hold for $t+1$. As consequence, we proved by induction that for all $t$, $V_{\beta}^t(.)$ is increasing with $a_i$. As $\underset{t \Rightarrow +\infty}{\lim} V_{\beta}^t(.)=V_{\beta}(.)$, then $V_{\beta}(.)$ is as well an increasing function with $a_i$.\\

Following the same procedure as done for the first case ($1-r\geq |p-r|$), we establish that the optimal solution of the bellman equation \eqref{eq:bellman_equation_discounted} is also a threshold based policy. 
Explicitly, there exists $a_n$ such that when the current state $a_{j} < a_n$, the prescribed action is a passive action, and when $ a_j \geq a_{n}$, the prescribed action is an active action.
That concludes the proof for the second case where $1-r < |p-r|$. 
\end{itemize}

\section{Proof of Theorem \ref{theo:threshold_policy_original_condition}}\label{app:theo:threshold_policy_original_condition}
We prove that there exits a constant $K$ such that for all $i$ and $\beta$, we have that:
\begin{equation}
|V_{\beta}(a_i)-V_{\beta}(a_0)| \leq K
\end{equation}
For that, we distinguish between two cases.
\begin{itemize}
\item $1-r \geq |p-r|$:
We consider a finite threshold $a_m(\beta)$:
\begin{itemize}
\item $i \geq m$:\\
We have the expected time to go from the state $a_i \geq a_m(\beta)$ to $a_0$ denoted by $E(T_{i0})$ is bounded by a given a constant that doesn't depend neither on $\beta$ nor on $i$.
Indeed $E(T_{i0})=\sum_{k=1}^{+\infty} k \rho (1-\rho)^{k-1}=M$.
Thus, 
\begin{align}
V_{\beta}(a_i)=&E[\sum_{t=0}^{T_{i0}-1} \beta^t (a(t)+\lambda d(t))|a(0)=a_i]\nonumber\\
&+E[\sum_{t=T_{i0}}^{+\infty} \beta^t (a(t)+\lambda d(t))|a(T_{i0})=a_0]\nonumber\\ 
\leq^{(a)}& (\underset{i \rightarrow +\infty}{\lim} a_i+\lambda) E(T_{i0})+E(\beta^{T_{i0}})V_{\beta}(a_0)\nonumber \\
\leq& (\underset{i \rightarrow +\infty}{\lim} a_i+\lambda)M+V_{\beta}(a_0)
\end{align}  
The inequality $(a)$ comes from the fact that as $a_i$ is increasing with $i$ when $1-r \geq |p-r|$, then $a_i$ is less than $\underset{i \rightarrow +\infty}{\lim} a_i$ for all integer $i$.
As consequence: $V_{\beta}(a_i)-V_{\beta}(a_0) \leq (l+\lambda)M$, where $l=\underset{i \rightarrow +\infty}{\lim} a_i$.\\
\item  $i<m$:\\
For $i\leq m$, $V_{\beta}(a_i)-V_{\beta}(a_0) \leq V_{\beta}(a_m)-V_{\beta}(a_0) \leq (l+\lambda)M$.\\
\end{itemize}
If the threshold policy is infinite. That means for all $i$ and $\beta$, $V_{\beta}(a_i)=\sum_{k=0}^{+\infty} \beta^k a_{i+k}$.
Consequently:
$$V_{\beta}(a_i)-V_{\beta}(a_0)=\sum_{k=0}^{+\infty} \beta^k (a_{i+k}-a_k)$$
$$|V_{\beta}(a_i)-V_{\beta}(a_0)| \leq \sum_{k=0}^{+\infty} |a_{i+k}-a_k|$$
We have:
$a_{i+k}-a_k=\sum_{j=k}^{i+k-1} a_{j+1}-a_j=\sum_{j=k}^{i+k-1} [(1-r)^{j+1}-(p-r)^{j+1}]$.\\
Therefore, by computing the sum, we have:
\begin{align}
|a_{i+k}-a_k| \leq &(1-r)^{k+1} \frac{1-(1-r)^i}{r}\nonumber \\
&+|p-r|^{k+1} \frac{1-|p-r|^i}{1-|p-r|}
\end{align}
Hence:
\begin{align}
\sum_{k=0}^{+\infty} |a_{i+k}-a_k| \leq& \frac{(1-(1-r)^i)(1-r)}{r^2} \nonumber \\
&+\frac{(1-|p-r|^i)|p-r|}{(1-|p-r|)^2}\nonumber \\
 \leq& \frac{(1-r)}{r^2}+\frac{|p-r|}{(1-|p-r|)^2}
\end{align}
Then: 
$$|V_{\beta}(a_i)-V_{\beta}(a_0)|\leq \frac{(1-r)}{r^2}+\frac{|p-r|}{(1-|p-r|)^2}$$
That concludes the result when the threshold is infinite.
\item $(1-r) <|p-r|$:
We distinguish between two cases:
\begin{itemize}
\item The threshold policy is strictly greater than $a_0$:\\
We have for all $i$ and $\beta$:
\begin{equation}\label{eq:inequality_difference_value_discounted_case_2}
V_{\beta}(a_i)-V_{\beta}(a_0) \leq V_{\beta}(a_1)-V_{\beta}(a_0)
\end{equation}
since $a_i \leq a_1$ for all $i$.
Given that $V_{\beta}(a_0)=a_0+\beta V_{\beta}(a_1)=\beta V_{\beta}(a_1)$, then: 
\begin{equation}\label{eq:expression_difference_value_discounted_case_2}
V_{\beta}(a_1)-V_{\beta}(a_0)=(1-\beta) V_{\beta}(a_1)
\end{equation}
To that extent, we need to show that $(1-\beta) V_{\beta}(a_1)$ is less than a constant independent of $\beta$.\\
To that end, in the sequel, we prove by induction on $t$ that for all $i$, $V_{\beta}^t(a_i) \leq \sum_{k=0}^{+\infty} \beta^k a_{i+k}$.\\
\begin{itemize}
\item $t=0$:\\
$V_{\beta}^0(a_i) \leq \sum_{k=0}^{+\infty} \beta^k a_{i+k}$ as $V^0_{\beta}(\cdot)=0$, then the statement holds for $t=0$.
\item $V_{\beta}^t(a_i) \leq \sum_{k=0}^{+\infty} \beta^k a_{i+k} \Rightarrow V_{\beta}^{t+1}(a_i) \leq \sum_{k=0}^{+\infty} \beta^k a_{i+k}$:\\
For that purpose, we suppose that $V_{\beta}^t(a_i) \leq \sum_{k=0}^{+\infty} \beta^k a_{i+k}$. Given that, $V_{\beta}^{t+1}(a_i)=\min \{a_i+\beta V^t_{\beta}(a_{i+1}), a_i+ \lambda+ \rho \beta V^t_{\beta}(a_0)+(1-\rho)\beta V_{\beta}^t(a_{i+1})\}$ subsequently:
\begin{align}
V_{\beta}^{t+1}(a_i) \leq& a_i+\beta V^t_{\beta}(a_{i+1})\nonumber \\
 \leq& a_i+\beta \sum_{k=0}^{+\infty} \beta^k a_{i+1+k}\nonumber\\
 =& a_i+\sum_{k=1}^{+\infty} \beta^k a_{i+k}\nonumber \\
 =&\sum_{k=0}^{+\infty} \beta^k a_{i+k}
\end{align}
That means the statement holds for $t+1$.
\end{itemize}
Thereby, we proved by induction that for all $t$, $V_{\beta}^t(a_i)\leq \sum_{k=0}^{+\infty} \beta^k a_{i+k}$. 
Knowing that $V_{\beta}$ is the limit of $V_{\beta}^t$, then $V_{\beta}(a_i) \leq \sum_{k=0}^{+\infty} \beta^k a_{i+k}$.\\
Accordingly, $V_{\beta}(a_1) \leq \sum_{k=0}^{+\infty} \beta^k a_{1+k}\leq \frac{a_1}{1-\beta}$ (since $a_i \leq a_1$ for all $i$).
Leveraging Equations \eqref{eq:inequality_difference_value_discounted_case_2} and \eqref{eq:expression_difference_value_discounted_case_2},we get:
$$|V_{\beta}(a_i)-V_{\beta}(a_0)| \leq |V_{\beta}(a_1)-V_{\beta}(a_0)| \leq a_1$$
Consequently we get our desired result.
\item The threshold policy is less than $a_0$:\\
We have  $V_{\beta}(a_0)=a_0+\rho V_{\beta}(a_0)+ (1-\rho)\beta V_{\beta}(a_1)$.
Then, $V_{\beta}(a_0)=\frac{\beta(1-\rho)}{1-\rho \beta}V_{\beta}(a_1)$.
That means, $V_{\beta}(a_1)-V_{\beta}(a_0) =\frac{1-\beta}{1-\rho \beta}V_{\beta}(a_1)$.
Following the same approach as the done for the first case, we obtain: $$|V_{\beta}(a_i)-V_{\beta}(a_0)| \leq \frac{a_1}{1-\rho \beta}\leq \frac{a_1}{1-\rho}$$
That concludes the result.
\end{itemize}
\end{itemize}

\section{Proof of Proposition \ref{prop:stationary_distribution}}\label{app:prop:stationary_distribution}
In order to demonstrate this proposition, we need to resolve the full balance equation under threshold policy $n$ at each state $a_{j}$:
\begin{equation}
u^n(a_{j})=\sum_{i=0}^{+\infty} pt^n(i \rightarrow j)u^n(a_{i})
\end{equation}
where $pt^n(i \rightarrow j)$ denotes the transitioning probability from the state $a_{i}$ to state $a_{j}$ under threshold policy $n$. 
After some computations, we obtain the desired result.  

\section{Proof of Theorem \ref{theo:expression_threshold_policy_in_function_lambda}}\label{app:theo:expression_threshold_policy_in_function_lambda}
In order to prove this proposition, we start first by showing that $\lambda(a_{n})$ is increasing with $a_{n}$. However, since $a_{n}$ is increasing with $n$ (Lemma \ref{lem:MAoII_increasing}), it is sufficient to show that $\lambda(a_{n})$ is increasing with $n$ to establish the desired result.  

Therefore, we first seek a closed-form expression of the intersection point $\lambda(a_{n})$, we obtain:
\begin{align}
\lambda(a_{n})=&\frac{(N-1)(N+1-N r) \rho}{N^2 r^2} \nonumber \\ 
&+(1-Nr)^{n+2}(n\rho+1+\frac{\rho}{N r}) \nonumber \\
& \ \ \ \times [\frac{1-(1-\rho)(1+(N-1)r)}{N r (1-(1-\rho)(1-r))} ]\nonumber \\
&-(1-r)^{n+2}(n\rho+1+\frac{\rho}{r}) \nonumber \\
& \ \  \ \times [\frac{\rho}{ r (1-(1-\rho)(1-r))}]
\end{align}
\begin{lemma}\label{lem:increasing_whittle_index_first_case}
The sequence $\lambda(a_{n})$ is strictly increasing with $n$.
\end{lemma}
\begin{IEEEproof}
One can see the \cite{kriouile2021minimizing}. 
\end{IEEEproof}
Now, we provide an useful lemma that allow us to establish our desired result.
Without loss of generality,  for notational convenience, we abbreviate $\lambda(a_n)$ by $\lambda(n)$.
\begin{lemma}\label{lem:relation_lambda}
For $\lambda \leq \lambda(k)$:
$$a^k+\lambda d^k \leq a^{k+1}+\lambda d^{k+1}$$
For $\lambda > \lambda(k)$: 
$$a^k+\lambda d^k > a^{k+1}+\lambda d^{k+1}$$ 
\end{lemma}
\begin{proof}
By definition of $\lambda(k)$, $a^k+\lambda(k) d^k = a^{k+1}+\lambda(k) d^{k+1}$.
Therefore, we have that:
\begin{align}
&(a^{k}+\lambda d^{k}) - (a^{k+1}+\lambda d^{k+1})\\
&=\lambda(k) d^{k+1}-\lambda(k) d^k +\lambda d^{k} - \lambda d^{k+1} \nonumber\\
&=(\lambda-\lambda(k))(d^k-d^{k+1})
\end{align} 
Given that $d^k$ is strictly decreasing with $n$, then if $\lambda \leq \lambda(k)$, $a^k+\lambda d^k \leq a^{k+1}+\lambda d^{k+1}$ and if $\lambda > \lambda(k)$, $a^k+\lambda d^k > a^{k+1}+\lambda d^{k+1}$.
The proof is complete.
\end{proof}  
Leveraging this above Lemma, we prove our desired results presented in these two following lemmas:
\begin{lemma}\label{lem:first_inequality_optimal_threshold}
For $\lambda > \lambda(n)$ and for $k \leq n$:
$$a^{n+1}+\lambda d^{n+1} < a^k+\lambda d^k$$
\end{lemma}
\begin{proof}
In order to prove this lemma, it is sufficient to show that $a^k+\lambda d^k$ is strictly decreasing with $k$ when $k \leq n$ and $\lambda > \lambda(n)$.
To that end, we prove that $a^{k+1}+\lambda d^{k+1} < a^k+\lambda d^k$ when $k \leq n$ and $\lambda > \lambda(n)$.
As $\lambda(n) \geq \lambda(k)$ according to Lemma \ref{lem:increasing_whittle_index_first_case}, then $\lambda > \lambda(k)$. Hence, according to Lemma \ref{lem:relation_lambda}, $a^{k+1}+\lambda d^{k+1} < a^k+\lambda d^k$. Thus, $a^k+\lambda d^k$ is strictly decreasing with $k$ when $k \leq n$ and $\lambda > \lambda(n)$. 
That means, $a^{n+1}+\lambda d^{n+1} < a^k+\lambda d^k$ for all $k \leq n$.  
\end{proof}
\begin{lemma}\label{lem:second_inequality_optimal_threshold}
For $\lambda \leq \lambda(n+1)$ and for $k > n$:
$$a^{n+1}+\lambda d^{n+1} \leq a^k+\lambda d^k$$
\end{lemma}
\begin{proof}
Likewise, we show that $a^k+\lambda d^k$ is increasing with $k$ when $k > n$ and $\lambda \leq \lambda(n+1)$.
As $\lambda(n+1) \leq \lambda(k)$ according to Lemma \ref{lem:increasing_whittle_index_first_case}, then $\lambda \leq \lambda(k)$. Hence, according to Lemma \ref{lem:relation_lambda}, $a^{k+1}+\lambda d^{k+1} \geq a^k+\lambda d^k$. Thus, $a^k+\lambda d^k$ is increasing with $k$ when $k > n$ and $\lambda \leq \lambda(n+1)$. 
That means, $a^{n+1}+\lambda d^{n+1} \leq a^k+\lambda d^k$ for all $k > n$. 
\end{proof}
Combining lemmas \ref{lem:first_inequality_optimal_threshold} and \ref{lem:second_inequality_optimal_threshold}, when $\lambda \in ]\lambda(n), \lambda(n+1)]$, we have $a^{n+1}+\lambda d^{n+1} \leq a^k+\lambda d^k$, for all $k \geq 0$.
Therefore, the optimal threshold when $\lambda \in ]\lambda(n), \lambda(n+1)]$ is $a_{n+1}$.
When $\lambda \leq \lambda(0)$, then $\lambda \leq \lambda(k)$ for all $k\geq 0$. Hence, according to Lemma \ref{lem:relation_lambda}, $a^{k}+\lambda d^{k} \leq a^{k+1}+\lambda d^{k+1}$ for all $k$. As consequence, the optimal threshold is $a_0$ since $a^{0}+\lambda d^{0} \leq a^{k}+\lambda d^{k}$ for all $k \geq 0$.\\  
Now, we deal with the third statement of the theorem.
We have $\underset{n \Rightarrow +\infty}{\lim} \lambda(a_n)=\frac{(N-1)(N+1-N r) \rho}{N^2 r^2}=\lambda_l$. To that extent, we prove that if $\lambda \geq \frac{(N-1)(N+1-N r) \rho}{N^2 r^2}$, then the optimal threshold policy is infinite:
\begin{lemma}
If $\lambda \geq \lambda_l$, then the optimal threshold policy is infinite.
\end{lemma} 
\begin{proof}
In order to establish our statement, we demonstrate that for $\lambda \geq \lambda_l$, we have that for all $n$, $a^{n+1}+\lambda d^{n+1} < a^n+ \lambda d^n$. By showing that, we establish that the optimal threshold denoted by $n^*$ could not be finite, otherwise $ a^{n^*+1}+ \lambda d^{n^*+1} \geq a^{n^*}+\lambda d^{n^*}$ which contradicts our claim.
To proceed so, we leverage the result in Lemma \ref{lem:increasing_whittle_index_first_case} that states that $\lambda(n)$ is increasing with $n$. Indeed, as $\lambda(n)$ is strictly increasing with $n$, then $\lambda(n) < \underset{n \Rightarrow +\infty}{\lim} \lambda(n)=\lambda_l$ for all $n$. Hence, according to Lemma \ref{lem:relation_lambda}, for $\lambda \geq \lambda_l > \lambda(n)$, we have that $a^{n+1}+\lambda d^{n+1} < a^n+ \lambda d^n$. That concludes the proof.   
\end{proof}

\section{Proof of Theorem \ref{theo:expression_threshold_policy_in_function_lambda_second_case}}\label{app:theo:expression_threshold_policy_in_function_lambda_second_case}
The proof follows the same lines of Theorem \ref{theo:expression_threshold_policy_in_function_lambda}. To that extent, for sake of brevity, we prove only that $\lambda(a_{2n})$ is increasing with $n$.
Therefore, we first seek a closed-form expression of the intersection point $\lambda(a_{2n})$, we obtain:
\begin{align}
\lambda(a_{2n})=&K_1\times \Big [[-(2-Nr)[1-(1-\rho)(1-Nr)] \nonumber \\
& \hspace{1cm}+2(1-Nr)^{2n+1}[1-(1-\rho)(1-Nr)^2]\nonumber \\
& \hspace{1cm}+(1-\rho)^{n+1}(1-Nr)^{2n+1}[(1-Nr)^2-1]\Big]\nonumber \\
-&K_2 \times \Big [[-(2-r)[1-(1-\rho)(1-r)]\nonumber \\
& \hspace{1cm}+2(1-r)^{2n+1}[1-(1-\rho)(1-r)^2]\nonumber \\
& \hspace{1cm}+(1-\rho)^{n+1}(1-r)^{2n+1}[(1-r)^2-1]\Big]\nonumber \\
\end{align}
where $K_1=\frac{\rho (p-r)}{Nr[1-(p-r)^2(1-\rho)][1-(p-r)(1-\rho)]}$ and $K_2=\frac{\rho (1-r)}{Nr[1-(1-r)^2(1-\rho)][1-(1-r)(1-\rho)]}$.

\begin{lemma}
The sequence $\lambda(a_{2n})$ is increasing with $n$.
\end{lemma}
\begin{IEEEproof}
After some mathematical analysis and algebraic manipulations, we get:
\begin{align}
\lambda(a_{2(n+1)})-\lambda(a_{2n})=K_1&\times [1-(1-\rho)(1-Nr)^2] \nonumber \\
&\times [1-(1-Nr)^{2}] \nonumber \\
&\times (1-Nr)^{2n+1} \nonumber\\
&\times [(1-\rho)^{n+1}-2] \nonumber \\
-K_2&\times [1-(1-\rho)(1-r)^2] \nonumber \\
&\times [1-(1-r)^{2}] \nonumber\\
&\times (1-r)^{2n+1} \nonumber \\
&\times [(1-\rho)^{n+1}-2] \\
\end{align}
Leveraging the equation above and given that $K_1 \leq 0$ (since $p\leq r$) and $K_2 \geq 0$, then $\lambda(a_{2(n+1)})-\lambda(a_{2n}) \geq 0$.
That concludes the proof. 
\end{IEEEproof}
As for the remaining points to prove, as we have emphasized in the beginning of this proof, the approach follows the same steps as done for the case $1-r \geq |p-r|$. For this reason, we omit the next points.

\section{Proof of Corollary \ref{cor:probability_G_T}}\label{app:cor:probability_G_T}
If we take $\epsilon=\sqrt{log(T)}$, according to Lemma \ref{lem:hoeffding_inequality}, then:
$$P(|r_i-r^*| \geq \sqrt{\frac{log(T)}{i}})\leq 2exp(-2(N-1)^2 log(T))\leq \frac{2}{T^2}$$
Therefore:
\begin{align}
&P(\underset{i=[1,\cdots,T]}{ \bigcup}|r_i-r^*| \geq \sqrt{\frac{log(T)}{i}}) \nonumber \\
&\leq \sum_{i=1}^{T} P(\{|r_i-r^*| \geq \sqrt{log(T)/i} \}) \nonumber \\
&\leq \frac{2}{T}
\end{align}
Hence: 
$$P(\underset{i=[1,\cdots,T]}{ \bigcap}|r_i-r^*| \leq \sqrt{\frac{log(T)}{i}})\geq 1- \frac{2}{T}$$
That concludes the proof.

\section{Proof of Proposition \ref{prop:probability_N_t_threshold_0}}\label{app:prop:probability_N_t_threshold_0}
We start first by giving an useful result. In fact, we apply the lemma \ref{lem:hoeffding_inequality} for the random variable $c_i$. By doing so, we get the following:
\begin{align}
&P(|\sum_{i=0}^{N(t)} c_i-(1+N(t))c| \geq \sqrt{(N(t)+1) log(T)}) \nonumber \\
&\leq 2\exp(-2 log(T))= \frac{2}{T^2}
\end{align}
That means: 
\begin{align}
&P(\sum_{i=0}^{N(t)} c_i-(1+N(t))c \leq \sqrt{(N(t)+1) log(T)}) \nonumber \\
& \geq P(|\sum_{i=0}^{N(t)} c_i-(1+N(t))c| \leq \sqrt{(N(t)+1) log(T)}) \nonumber \\
&\geq 1-\frac{2}{T^2}
\end{align}
Hence, we have: $\sum_{i=0}^{N(t)} c_i \leq \sqrt{(N(t)+1) log(T)}+(1+N(t))c$ with at least p.b $1-\frac{2}{T^2}$.
As the threshold from time $0$ till $T_0-1$ is equal to $0$ then, we have $t_{N(t)+1}=\sum_{i=0}^{N(t)} c_i$. Given that, by definition, $t_{N(t)}\leq t \leq t_{N(t)+1}$, then:
$$t\leq t_{N(t)+1}=\sum_{i=0}^{N(t)}c_i < \sqrt{(N(t)+1) log(T)}+(N(t)+1)c$$
with p.b at least $1-\frac{2}{T^2}$
\begin{lemma}
$$t< \sqrt{(N(t)+1) log(T)}+(N(t)+1)c$$
implies 
$$N(t) > \frac{1}{c^2}(\sqrt{tc}(\sqrt{tc}-\sqrt{log(T)})-1$$
\end{lemma}
\begin{proof}
We consider the function $f(x)=\sqrt{log(T)}x+x^2c-t$ where $x$ is the unknown variable. We find for which $x \geq 0$, $f(x) > 0$. By resolving the equation of the second order $f(x)=0$, we find that for $x > \frac{\sqrt{log(T)+4tc}-\sqrt{log(T)}}{2c}$, $f(x) > 0$. 
Therefore replacing $x$ by $\sqrt{1+N(t)}$, we get:
$$\sqrt{log(T)}\sqrt{1+N(t)}+(1+N(t))c-t >0$$
$$ \Rightarrow \sqrt{1+N(t)} > \frac{\sqrt{log(T)+4tc}-\sqrt{log(T)}}{2c}$$
Thus, if $t< \sqrt{(N(t)+1) log(T)}+(N(t)+1)c$:
\begin{align}
 1+N(t) &> \frac{1}{4c^2}(2log(T)+4tc-2\sqrt{log(T)}\sqrt{log(T)+4tc}) \nonumber\\
& \geq \frac{1}{4c^2}(2log(T)+4tc-2log(T)-2\sqrt{4log(T)tc}) \nonumber \\
&= \frac{1}{4c^2}(4tc-4\sqrt{log(T)tc})
\end{align} 
As consequence: 
$$N(t) > \frac{1}{c^2}(tc-\sqrt{log(T)tc})-1$$
\end{proof}
That means:
$$N(t) > \frac{1}{c^2}(\sqrt{tc}(\sqrt{tc}-\sqrt{log(T)})-1$$
with p.b at least $1-\frac{2}{T^2}$.
Hence, the proof is complete.

\section{Proof of Proposition \ref{prop:T_0_less_than_L_0}}\label{app:prop:T_0_less_than_L_0}
We prove our proposition by contradiction. To that extent, we consider that $T_0 >L_0$ and we prove that for $t=L_0$, knowing $G(t)$ and $H(t)$, we have that $r_{N(L_0)-1}+\sqrt{log(T)/(N(L_0)-1)} < r_l$ or
$r_{N(L_0)-1}-\sqrt{log(T)/(N(L_0)-1)} \geq r_l$.  
\begin{lemma}\label{lem:inequality_N_L_0_knowing_H_t}
knowing $H(L_0)$, we have for $t=L_0$, $N(L_0)-1 > \frac{4log(T)}{|r^*-r_l|^2}$
\end{lemma}
\begin{proof}
Through this proof we consider that $t=L_0$.
We have by definition of $L_0$:
\begin{align}
&t >log(T)/c \Big (1+c\sqrt{\frac{4}{|r^*-r_l|^2}+2} \Big )^2 \\
\Rightarrow & \sqrt{tc} >\sqrt{log(T)}(1+c\sqrt{\frac{4}{|r^*-r_l|^2}+2})\\
\Rightarrow & \sqrt{tc}-\sqrt{log(T)} >\sqrt{log(T)}c\sqrt{\frac{4}{|r^*-r_l|^2}+2})
\end{align}
On the other hand, 
\begin{equation}
\sqrt{tc} >\sqrt{log(T)}c\sqrt{\frac{4}{|r^*-r_l|^2}+2})
\end{equation}
Therefore, the product $\sqrt{tc}(\sqrt{tc}-\sqrt{log(T)})$ is greater than $c^2 log(T)(\frac{4}{|r^*-r_l|^2}+2)$. Thus:
\begin{align}
\frac{1}{c^2}\sqrt{tc}(\sqrt{tc}-\sqrt{log(T)}) &>  log(T)\frac{4}{|r^*-r_l|^2}+2log(T)\nonumber\\
&> log(T)\frac{4}{|r^*-r_l|^2}+2
\end{align} 
Consequently:
\begin{equation}
 \frac{1}{c^2}\sqrt{tc}(\sqrt{tc}-\sqrt{log(T)})-1 >\frac{4log(T)}{|r^*-r_l|^2}+1
 \end{equation} 
Baring in mind the event $H(t)$, then $N(t)-1 > \frac{4log(T)}{|r^*-r_l|^2}$. Thereby, $N(L_0)-1 > \frac{4log(T)}{|r^*-r_l|^2}$.       
\end{proof}
Leveraging the lemma above, we establish our proposition. To that end, we provide the following lemma.
\begin{lemma}\label{lem:condition_for_t_knowing_G_for_event_of_chosing}
Knowing $G(T)$, if we have for a given $t$ such that $|r^*-r_l| > 2\sqrt{\frac{log(T)}{N(t)-1}}$, then: 
$$r_{N(t)-1}+\sqrt{log(T)/(N(t)-1)} < r_l$$
or
$$r_{N(t)-1}-\sqrt{log(T)/(N(t)-1)} > r_l$$
\end{lemma}
\begin{proof}
$|r^*-r_l| > 2\sqrt{\frac{log(T)}{N(t)-1}}$ implies that:
$$r^*+2\sqrt{log(T)/(N(t)-1)} < r_l$$
or 
$$r^*-2\sqrt{log(T)/(N(t)-1)} > r_l$$
Given $G(T)$, then: 
$$r^*-\sqrt{log(T)/(N(t)-1)} \leq r_{N(t)-1} \leq r^*+\sqrt{log(T)/(N(t)-1)}$$ 
Therefore: $$r_{N(t)-1}+\sqrt{log(T)/(N(t)-1)} \leq r^*+2\sqrt{log(T)/(N(t)-1)}$$
and  
$$r_{N(t)-1}-\sqrt{log(T)/(N(t)-1)} \geq r^*-2\sqrt{log(T)/(N(t)-1)}$$
As consequence, either: 
$$r_{N(t)-1}+\sqrt{log(T)/(N(t)-1)} <r_l$$
or 
$$r_{N(t)-1}-\sqrt{log(T)/(N(t)-1)} > r_l$$
That concludes the proof.
\end{proof}
The above lemma gives us a sufficient condition on $t$ such that the event $\{r_{N(t)-1}+\sqrt{log(T)/(N(t)-1)} \leq r_l\} \cup \{r_{N(t)-1}-\sqrt{log(T)/(N(t)-1)} > r_l\}$ is satisfied. To that extent, we need to check if $t=L_0$ satisfies this condition. Indeed, knowing $H(L_0)$, according to Lemma \ref{lem:inequality_N_L_0_knowing_H_t}, we have $N(L_0)-1 > \frac{4log(T)}{|r^*-r_l|^2}$, that is, $|r^*-r_l| > 2\sqrt{\frac{log(T)}{N(L_0)-1}}$. Hence the condition given in Lemma \ref{lem:condition_for_t_knowing_G_for_event_of_chosing} is satisfied for $t=L_0$.
That means, according to the same lemma, we have, knowing $G(T)$, $$r_{N(L_0)-1}+\sqrt{log(T)/(N(L_0)-1)} < r_l$$ 
or
$$r_{N(L_0)-1}-\sqrt{log(T)/(N(L_0)-1)} \geq r_l$$
To sum up, if $G(T)$ and $H(t)$ are realized for $t \in [0,T_0-1]$, and considering that $L_0<T_0$, we find that when $t=L_0$, we have that $r_{N(L_0)-1}+\sqrt{log(T)/(N(L_0)-1)} < r_l$ or $r_{N(L_0)-1}-\sqrt{log(T)/(N(L_0)-1)} \geq r_l$. That contradicts with the fact that $T_0$ is the first time-stamp such that $r_{N(T_0)-1}+\sqrt{log(T_0)/(N(T_0)-1)} < r_l$ or $r_{N(T_0)-1}-\sqrt{log(T)/(N(T_0)-1)} \geq r_l$. Therefore, $T_0$ must be less than $L_0$.     

\section{Proof of Proposition \ref{prop:probability_N_t_threshold_n_r}}\label{app:prop:probability_N_t_threshold_n_r}
We have for all $t\geq T_1$:
$$t_{N(t)+1}-T_1=\sum_{i=N(T_1)-1}^{N(t)-1}n(r_i)+\sum_{i=N(T_1)}^{N(t)}c_i$$
Hence, knowing $G(T)$ and $A(T_0)$:
$$t-T_1 \leq (N(t)-N(T_1)+1)m+\sum_{i=N(T_1)}^{N(t)}c_i$$
i.e.
$$t' \leq (N'(t)+1)m+\sum_{i=N(T_1)}^{N(t)}c_i$$
Using Lemma \ref{lem:hoeffding_inequality}, we have that $\sum_{i=N(T_1)}^{N(t)}c_i$ is less than $\sqrt{(N'(t)+1) log(T)}+(N'(t)+1)c$ with at least probability $1-\frac{2}{T^2}$.
Therefore, we have:
$$t' < \sqrt{(N'(t)+1) log(T)}+(N'(t)+1)(m+c)$$
with p.b $1-\frac{2}{T^2}$.
Following the same procedure as done for Proposition \ref{prop:probability_N_t_threshold_0}, we have that:
$$t' < \sqrt{(N'(t)+1) log(T)}+(N'(t)+1)(m+c)$$
implies that:
$$N'(t) > \frac{1}{(m+c)^2}(\sqrt{t'(m+c)}(\sqrt{t'(m+c)}-\sqrt{log(T)})-1$$
As consequence:
$$P(N'(t) > \frac{1}{(m+c)^2}(\sqrt{t'(m+c)}(\sqrt{t'(m+c)}-\sqrt{log(T)})-1$$
$$|G(T),A(T_0))\geq 1-\frac{2}{T^2}$$

\section{Proof of Proposition \ref{prop:similarity_prob_pi_and_optim_thresh}}\label{app:prop:similarity_prob_pi_and_optim_thresh}
We start first by providing these two following lemmas:
\begin{lemma}\label{lem:upper_bound_probability_fixed_N_t}
There exists a constant $b$ such that for $t\geq T_1$, we have that:
\begin{align}
P(\pi(s^{\pi}(t)) \neq \pi^*(s^{\pi}(t))&|G(T), A(T_0),N(t)=j)\nonumber \\
&\leq 2 exp(-2b^2 (j-1))
\end{align}
\end{lemma}
\begin{proof}
See appendix \ref{app:lem:upper_bound_probability_fixed_N_t} 
\end{proof}
\begin{lemma}\label{lem:upper_bound_N_t}
For $t \geq T_1$, there exists a constant $C_2$ such that if $t \geq C_2log(T)=L_2$, then knowing $W(t)$, we have that:
$$N'(t) \geq \frac{log(T)}{b^2}+1$$  
\end{lemma}
\begin{proof}
See appendix \ref{app:lem:upper_bound_N_t}.
\end{proof}
Using these above lemmas, we obtain:
\begin{align}
&P(F(L_2))\leq \sum_{t=L_2}^T P(\pi(s^{\pi}(t)) \neq \pi^*(s^{\pi}(t))|G(T),A(T_0),W(t)) \nonumber \\
&\leq \sum_{t=L_2}^T \sum_{j=0}^t  P(\pi(s^{\pi}(t) \neq \pi^*(s^{\pi}(t))|G(T), A(T_0),N(t)=j) \nonumber \\
&\hspace{3cm} \times P(N(t)=j|G(T), A(T_0),W(t))\nonumber \\
&\leq \sum_{t=L_2}^T \sum_{j=0}^t 2 exp(-2b^2 (j-1))P(N(t)=j|G(T), A(T_0),W(t))\nonumber \\
&\leq \sum_{t=L_2}^T 2 E_{N(t)}[exp(-2b^2 (N(t)-1)|G(t),A(T_0),W(t)]\nonumber \\
&= \sum_{t=L_2}^T 2 E_{N(t)}[exp(-2b^2 (N'(t)+N(T_1)-1)|G(t),A(T_0),W(t)]\nonumber \\
&\leq \sum_{t=T_1}^T \frac{2}{T^2} \nonumber \\
&\leq \frac{2}{T}
\end{align}

\section{Proof of Lemma \ref{lem:upper_bound_probability_fixed_N_t}}\label{app:lem:upper_bound_probability_fixed_N_t}
We have that:
\begin{align}
&P(\pi(s^{\pi}(t) \neq \pi^*(s^{\pi}(t))|G(T), A(T_0),N(t)=j) \nonumber \\
&=P(\pi(s^{\pi}(t) \neq \pi^*(s^{\pi}(t))| \pi^*(s^{\pi}(t))=0,G(T), A(T_0),N(t)=j)\nonumber \\
& \ \ \ \ \times P(\pi^*(s^{\pi}(t))=0|G(T), A(T_0),N(t)=j)\nonumber \\
&+P(\pi(s^{\pi}(t) \neq \pi^*(s^{\pi}(t))| \pi^*(s^{\pi}(t))=1,G(T), A(T_0),N(t)=j) \nonumber \\
& \ \ \ \ \times P(\pi^*(s^{\pi}(t))=1|G(T), A(T_0),N(t)=j)
\end{align}
To that extent, we bound $P(\pi(s^{\pi}(t) \neq \pi^*(s^{\pi}(t))| \pi^*(s^{\pi}(t))=0,G(T), A(T_0),N(t)=j)$ and $P(\pi(s^{\pi}(t) \neq \pi^*(s^{\pi}(t))| \pi^*(s^{\pi}(t))=1,G(T), A(T_0),N(t)=j)$
\begin{itemize}
\item $P(\pi(s^{\pi}(t) \neq \pi^*(s^{\pi}(t))| \pi^*(s^{\pi}(t))=0,G(T), A(T_0),N(t)=j)$: \\
$\pi^*(s^{\pi}(t))=0$ is equivalent to $n(r^*) > s^{\pi}(t)$. In this case, $\pi(s^{\pi}(t)) \neq \pi^*(s^{\pi}(t))$ implies that $\pi(s(t))=1$. Hence, if $\pi^*(s^{\pi}(t))=0$, $\pi(s^{\pi}(t)) \neq \pi^*(s^{\pi}(t))$ implies that $n(t)\leq s^{\pi}(t)$. To that extent, in the sequel we compute the probability $P(n(t) \leq s^{\pi}(t) |n(r^*) > s^{\pi}(t), G(T), A(T_0),N(t)=j)$.
\begin{align}
&P(n(t) \leq s^{\pi}(t) |n(r^*) > s^{\pi}(t), G(T), A(T_0),N(t)=j)\nonumber\\
&=P(n(t)-n(r^*) \leq s^{\pi}(t)-n(r^*) \nonumber \\
&\hspace{2cm} |n(r^*) > s^{\pi}(t), G(T), A(T_0),N(t)=j) \nonumber \\
&=P(n(t)-n(r^*) \leq -(n(r^*)-s^{\pi}(t)) \nonumber \\
& \hspace{2cm} |n(r^*) > s^{\pi}(t), G(T), A(T_0),N(t)=j) \nonumber \\
&\leq P(n(t)-n(r^*) \leq -1  \nonumber \\
& \hspace{2cm}|n(r^*) > s^{\pi}(t), G(T), A(T_0),N(t)=j) \nonumber \\
&\leq P(|n(t)-n(r^*)| \geq 1  \nonumber \\
& \hspace{2cm}|n(r^*) > s^{\pi}(t), G(T), A(T_0),N(t)=j)
\end{align}

\item $P(\pi(s^{\pi}(t) \neq \pi^*(s^{\pi}(t))| \pi^*(s^{\pi}(t))=1,G(T), A(T_0),N(t)=j)$:\\
Following the same analysis as for the first case, we get:
\begin{align}
&P(\pi(s^{\pi}(t) \neq \pi^*(s^{\pi}(t))\nonumber \\
&\hspace{2cm} | n(r^*) \leq s^{\pi}(t), G(T), A(T_0),N(t)=j) \nonumber \\
&\leq P(|n(t)-n(r^*)| \geq 1 \nonumber \\
&\hspace{2cm} |n(r^*) \leq s^{\pi}(t), G(T), A(T_0),N(t)=j)
\end{align}
\end{itemize}
We have $n(t)=n(t_{N(t)})$. Given that at $t_{N(t)}$, according to Algorithm \ref{euclid2}, we dispose of the estimation of $r^*$ at time $N(t)-1$, then, the optimal threshold computed at $t_{N(t)}$ is $n(r_{N(t)-1})$.
As the function $r \rightarrow n(r)$ is lipchitz when $t \geq T_1$ knowing $G(T)$ and $A(T_0)$, then there exist $b >0$ such that if $|n(t)-n(r^*)| \geq 1$ then $|r_{N(t)-1}-r^*|\geq b$.
That means, applying Lemma \ref{lem:hoeffding_inequality}:  
$$P(|n(t)-n(r^*)| \geq 1|n(r^*) > s^{\pi}(t), G(T), A(T_0),N(t)=j)$$
$$\leq 2 exp(-2b^2 (j-1))$$ 
and 
$$P(|n(t)-n(r^*)| \geq 1|n(r^*) \leq s^{\pi}(t), G(T), A(T_0),N(t)=j)$$
$$\leq 2 exp(-2b^2 (j-1))$$ 
Consequently, for $t\geq T_1$:
$$P(\pi(s^{\pi}(t) \neq \pi^*(s^{\pi}(t))|G(T), A(T_0),N(t)=j)$$
$$\leq 2 exp(-2b^2 (j-1))$$

\section{Proof of Lemma \ref{lem:upper_bound_N_t}}\label{app:lem:upper_bound_N_t}
We define the constant $C_2$ by $\frac{1}{m+c}[\frac{(m+c)^2}{b^2}+3+\frac{2(m+c)}{b}+4(m+c)]$, then:
$$t\geq C_2 log(T)=L_2$$
implies that:
\begin{align}
t\geq \frac{1}{m+c}&[(m+c)^2\frac{log(T)}{b^2}+2+log(T) \nonumber \\
&+2(m+c)\sqrt{\frac{log(T)^2}{b^2}+2log(T)}]
\end{align}
That is:
$$t(m+c) \geq [(m+c)\sqrt{\frac{log(T)}{b^2}+2}+\sqrt{log(T)}]^2$$
Thus,
$$\sqrt{t(m+c)}-\sqrt{log(T)} \geq (m+c)\sqrt{\frac{log(T)}{b^2}+2}$$
Consequently,
$$\frac{1}{m+c}\sqrt{t(m+c)}-\sqrt{log(T)} \geq \sqrt{\frac{log(T)}{b^2}+2}$$
On the other hand, we have:
$$\frac{1}{m+c}\sqrt{t(m+c)} \geq \sqrt{\frac{log(T)}{b^2}+2}$$
Therefore:
$$\frac{1}{(m+c)^2}(\sqrt{t(m+c)}-\sqrt{log(T)})\sqrt{t(m+c)} \geq \frac{log(T)}{b^2}+2$$
Knowing $W(t)$, we have that: $N'(t) \geq \frac{1}{(m+c)^2}(\sqrt{t(m+c)}-\sqrt{log(T)})\sqrt{t(m+c)} -1$. Hence, given $W(t)$, we have if $t \geq L_2$, $N'(t) \geq \frac{log(T)}{b^2}+1$. That concludes the proof.

\section{Proof of Proposition \ref{prop:spectral_value_less_one}}\label{app:prop:spectral_value_less_one}
We compute the polynomial characteristic of $Q$. 
In fact, after computations, we get: 
\begin{equation}
\chi_{Q}(\lambda)=(-1)^{n^*} \lambda^{n^*} -(-1)^{n^*+1}\rho\sum_{i=0}^{n^*-1} \lambda^i
\end{equation}
Now, based on the expression of the characteristic polynomial $\chi_{Q}$, we prove that all the eigenvalues of $Q$ have a modulus strictly less than one. We prove this result by contradiction. More specifically, we suppose there exists a given eigenvalue of the matrix $Q$ that satisfies $|\lambda|\geq 1$. As $\lambda$ is an eigenvalue of $Q$, it is therefore a root of $\chi_{Q}(\lambda)$. Hence, it verifies:
\begin{equation}
\lambda^{n^*}=-\rho \frac{1-\lambda^{n^*}}{1-\lambda}
\end{equation}
By factorizing the element $\lambda^{n^*}$, and by using the modulus on both sides, we get:
\begin{align*}
\rho&=|\lambda|^{n^*}|\rho-1+\lambda|\\
&\geq^{(a)} |\lambda|^{n^*}(|\lambda|-|1-\rho|)\\
& \geq^{(b)} |\lambda|(|\lambda|-|1-\rho|)\\
& = |\lambda|^2- |\lambda|(1-\rho)
\end{align*}
where $(a)$ and $(b)$ originate from the reverse triangular inequality and the fact that $|\lambda|\geq 1$ respectively. Hence:
\begin{equation}
|\lambda|^2- |\lambda|(1-\rho)-\rho \leq 0
\end{equation}
By employing standard real functions analysis, it can be shown that the polynomial:
\begin{equation}
x^2-(1-\rho)x-\rho
\end{equation}
is negative if and only if $x\in [-\rho,1]$. However, $|\lambda| \geq 1$ by assumption. Accordingly, $|\lambda|$ can only be equal to $1$. Next, we prove that, in this case, the imaginary part of $\lambda$ is equal to zero. To that end, let us consider $\lambda=x+iy$. Therefore, we have:
\begin{equation}
\rho=|\lambda|^{n^*} |\rho-1+x+iy|=|\rho-1+x+iy|
\end{equation}
By using the definition of the modulus, and by squaring both sides, we get:
\begin{equation}
\rho^2=(1-x-\rho)^2+y^2
\end{equation}
Knowing that $x^2+y^2=1$, we can deduce: 
\begin{equation}
\frac{2-2\rho}{2(1-\rho)}=x
\end{equation}  
Hence, $x=1$, i.e. $y=0$, and we can deduce that $\lambda=1$. 
However, $1$ is not eigenvalue of matrix $Q$. This can be seen by replacing $\lambda$ with $1$ in the characteristic polynomial of $Q$. Accordingly, the hypothesis that $|\lambda| \geq 1$ fails and all the eigenvalues of $Q$ have a modulus strictly less than $1$. Hence $\gamma <1$

\end{appendices}

\end{document}